\documentclass[12pt]{article}

\usepackage{amsmath}
\usepackage{amssymb}
\usepackage{amsthm}
\usepackage{hyperref}
\usepackage{enumerate}
\usepackage{bbm}
\usepackage{color}
\newtheorem{thm}{Theorem}[section]

\newtheorem{lem}[thm]{Lemma}
\newtheorem{conj}[thm]{Conjecture}
\newtheorem{assumption}[thm]{Assumption}

\newtheorem{pr}[thm]{Proposition}

\newtheorem{definition}[thm]{Definition}

\newtheorem{example}[thm]{Example}

\newtheorem{remark}[thm]{Remark}
\newenvironment{rem}{\begin{remark}\rm}{\end{remark}}

\newcommand{\E}{\mathbb{E}}
\newcommand{\R}{\mathbb R}
\newcommand{\eps}{\varepsilon}

\title{Extensions of the I-MMSE Relationship to Gaussian Channels with Feedback and Memory~\footnote{A preliminary version of this paper has been presented in IEEE ISIT 2015~\cite{HanSong2014}.}}

\author{\small \begin{tabular}{ccc}
Guangyue Han & Jian Song\\
The University of Hong Kong & The University of Hong Kong\\
email:  ghan@hku.hk & email: txjsong@hku.hk \\
\end{tabular}}

\date{\today}

\begin{document} \maketitle

\begin{abstract}
Unveiling a fundamental link between information theory and estimation theory, the I-MMSE relationship by Guo, Shamai and Verdu~\cite{gu05}, together with its numerous extensions, has great theoretical significance and various practical applications. On the other hand, its influences to date have been restricted to channels without feedback or memory, due to the absence of its extensions to such channels. In this paper, we propose extensions of the I-MMSE relationship to discrete-time and continuous-time Gaussian channels with feedback and/or memory. Our approach is based on a very simple observation, which can be applied to other scenarios, such as a simple and direct proof of the classical de Bruijn's identity.
\end{abstract}


\section{Introduction}

Consider the following discrete-time memoryless Gaussian channel
\begin{equation} \label{discrete-time-Gaussian-channel}
Y=\sqrt{snr} X+Z,
\end{equation}
where $snr$ denotes the signal-to-noise ratio (SNR) of the channel, $X$ and $Y$ denote the input and output of the channel, respectively, and the standard normally distributed noise $Z$ is independent of $X$. An interesting recent result by Guo, Shamai and Verdu~\cite{gu05} states that for any channel input $X$ with $E[X^2] < \infty$,
\begin{equation} \label{I-MMSE}
\frac{d}{d snr} I(X; Y)= \frac{1}{2} \E[(X-\E[X|Y])^2],
\end{equation}
where the left hand side is the derivative of $I(X; Y)$ with respect to $snr$, and the right-hand side is half of the so-called \emph{minimum mean-square error} (MMSE), which corresponds to   the best estimation of $X$ given the observation $Y$ in the mean-square error sense. The I-MMSE relationship (\ref{I-MMSE}) carries over verbatim to linear vector Gaussian channels and has been widely extended to continuous-time Gaussian channels~\cite{gu05}, general additive Gaussian channels~\cite{za05}, additive non-Gaussian channels~\cite{gu05a}, arbitrary channels~\cite{pa07}, derivatives with respect to arbitrary parameterizations~\cite{pa06}, higher order derivatives~\cite{pa09, le15}, and so on.

Unveiling an important link between information theory and estimation theory, the I-MMSE relationship as above and its numerous extensions are of fundamental significance to relevant areas in these two fields and have been exerting far-reaching influences over a wide-range of topics. Representative applications include, but not limited to, power allocation of parallel Gaussian channels~\cite{lo06}, analysis of extrinsic information of code ensembles~\cite{pe06},
Gaussian broadcast channels~\cite{gu11, li14}, Gaussian wiretap channels~\cite{gu11, bu09}, Gaussian interference channels~\cite{bu10, wu11}, a simple proof of the classical entropy power inequality~\cite{ve06}. For comprehensive references to the applications of the I-MMSE relationship and its extensions, we refer to~\cite{sh11, gu13} .

On the other hand, all the applications of the I-MMSE relationship to date have been restricted to channels without feedback or memory, due to the lack of extensions of the I-MMSE relationship to such channels. In this regard, a ``plain'' generalization of the original I-MMSE relationship to feedback channels should not be expected, which has been noted in~\cite{gu05}, where an example is given to show that the exact I-MMSE relationship fails to hold for some continuous-time feedback channel. In this paper, we remedy the situations with some explicit correctional terms (which vanish if the channel does not have feedback or memory) and extend the I-MMSE relationship to channels with feedback or memory.  Despite the fact that the I-MMSE relationship have been examined from a number of perspectives (see its multiple proofs in~\cite{gu05}), our approach is still novel and powerful. As a matter of fact, other than recovering and extending the I-MMSE relationship, our approach can be applied elsewhere, such as yielding a simple and direct proof of the classical de Bruijn's identity~\cite{st59, co85}; see Section~\ref{new-proof-2}.

Our approach is based on a surprisingly simple idea, which can be roughly stated as follows: before taking derivative of an information-theoretic quantity with respect to certain parameters, we represent it as an expectation with respect to a probability space independent of the parameters. For illustrative purpose, in what follows, we consider the discrete-time Gaussian channel in (\ref{discrete-time-Gaussian-channel}) and review a ``conventional'' proof of (\ref{I-MMSE}) in~\cite{gu05} and compare it with ours.

First, note that for the channel in (\ref{discrete-time-Gaussian-channel}), taking derivative of $I(X; Y)$ is equivalent to that of $H(Y)$, which can be written as the expectation of $-\log f_Y(Y)$:
\begin{equation*}
H(Y)=-\E[\log f_Y(Y)].
\end{equation*}
In their fourth proof of (\ref{I-MMSE}), the authors of~\cite{gu05} choose the probability space, with respect to which the expectation as above is taken, to be the sample space of $Y$ (with naturally induced measure), which obviously depends on $snr$. With respect to this probability space, $H(Y)$ is naturally expressed as:
\begin{equation*}
H(Y)=-\int_{\mathbb{R}} f_Y(y) \log f_Y(y) dy.
\end{equation*}
Then, under some mild assumptions, the derivative of $H(Y)$ with respect to $snr$ can penetrate into the integral, and then (\ref{I-MMSE}) follows from integration by parts and other straightforward computations.

Under our approach, we would rather choose a probability space independent of $snr$. For example, choosing the probability space to be the sample space of $(X, Z)$, we will express $H(Y)$ as
\begin{equation*}
H(Y)= -\int_{\mathbb{R}} \int_{\mathbb{R}} f_X(x) f_Z(z) \log f_Y(\sqrt{snr} x+z) dx dz.
\end{equation*}
It turns out such a seemingly innocent shift of viewpoint will render the follow-up computations rather simple and direct before reaching (\ref{I-MMSE}); and most importantly, when applied to channels with feedback and/or memory, it naturally leads to extensions of the I-MMSE relationship. For instance, consider the discrete-time Gaussian channel with feedback:
\begin{equation*}
Y_i=\sqrt{snr} X_i(M, Y_1^{i-1})+Z_i, \quad i=1, 2, \ldots, n,
\end{equation*}
where the channel input $X_i$ depends on the message $M$ and the previous channel outputs $Y_1^{i-1}$. Using the above-mentioned approach, we will obtain the following extension (see Remark~\ref{extension-feedback}) of the I-MMSE relationship:
\begin{align} \label{main-result-in-introduction}
\hspace{-1cm} \frac{d}{d snr} I(X_1^n \rightarrow Y_1^n) & = \frac{1}{2} \sum_{i=1}^n \E[(X_i-\E\left[X_i|Y_1^n])^2 \right] + snr \sum_{i=1}^n \E\left[X_i \E_W\left[\left.\frac{d}{d snr} X_i \right|Y_1^n \right]-\E[X_i|Y_1^n] \frac{d}{d snr} X_i\right] \notag\\
&\hspace{2cm}+\sqrt{snr} \sum_{i=1}^n \E\left[Y_i \left(\frac{d}{d snr} X_i-\E_W\left[\left.\frac{d}{d snr} X_i \right|Y_1^n\right] \right) \right],
\end{align}
where $X_i$ is the abbreviated form of $X_i(M, Y_1^{i-1})$ and $I(X_1^n \rightarrow Y_1^n)$ is the directed information~\cite{massey90} between $X_1^n$ and $Y_1^n$. Directed information is a notion generalized from mutual information for feedback channels, and the second and third terms at the right hand side of (\ref{main-result-in-introduction}) are  correctional terms, which vanish when $X_i$ does not depend on $Y_1^{i-1}$ ({\it i.e.}, there is no feedback), so (\ref{main-result-in-introduction}) is indeed an extension of the I-MMSE relationship in (\ref{I-MMSE}) to discrete-time Gaussian channels with feedback. As elaborated later, the I-MMSE relationship can be extended to Gaussian channels with feedback and/or memory, in either discrete-time or continuous-time.

The remainder of the paper is organized as follows. In Section~\ref{new-proofs}, based on the proposed approach, we give a new proof of the I-MMSE relationship for discrete-time Gaussian channels, and a new proof of the classical de Bruijn's identity. We will present our extensions of the I-MMSE relationship, the main results in this paper, in Sections~\ref{main-results-1} and~\ref{main-results-2}, which will be followed by an outlook for some promising future directions in Section~\ref{outlook}.

\section{New Proofs of Existing Results} \label{new-proofs}

In this section, to further illustrate the idea of our approach, we give new proofs of some existing results: the original I-MMSE relationship in (\ref{I-MMSE}) and the classical de Bruijn's identity. To enhance the readability and emphasize the main idea, here and throughout the paper, we omit some technical details of checking the conditions required for the interchanges of differentiation and integration, which will be provided in the Appendices.

\subsection{A new proof of the I-MMSE relationship} \label{new-proof-1}

In this section, we consider the Gaussian channel specified in (\ref{discrete-time-Gaussian-channel}) and give a new proof of (\ref{I-MMSE}). Here and throughout the paper, we replace $\sqrt{snr}$ with $\rho \in \mathbb{R}$ ($\rho$ can be negative) to avoid notational cumbersomeness during the computation; the derivative with respect to $snr$ can be readily obtained with an application of the chain rule. Then, with the above-mentioned replacement, the channel (\ref{discrete-time-Gaussian-channel}) becomes
\begin{equation*}
Y=\rho X+Z,
\end{equation*}
where $\rho \in \mathbb{R}$, and we only have to prove that
\begin{equation} \label{rho-I-MMSE}
\frac{d}{d \rho} I(X; Y)= \rho \E[(X-\E[X|Y])^2].
\end{equation}

Obviously, the conditional density of $Y$ given $X=x$ by $f_{Y|X}(y|x)=\frac{1}{\sqrt{2 \pi}} e^{-(y-\rho x)^2/2}$, and the density function of $Y$ can be computed as
\begin{equation*}
f_{Y}(y)=\int_{\R} f_{Y|X}(y|x) f_X(x) dx.
\end{equation*}
It follows from the assumption that $X$ is independent of $Z$ that
\begin{equation*}
I(X; Y)=H(Y)-H(Y|X)=H(Y)-H(Z|X)=H(Y)-H(Z),
\end{equation*}
which, together with the fact that $Z$ does not depend on $\rho$, implies that
\begin{equation*}
\frac{d}{d\rho} I(X; Y)=-\frac{d}{d\rho}\E\left[ \log f_{Y}(Y) \right]\stackrel{(a)}{=}-\E\left[\frac{d}{d\rho} \log f_{Y}(Y)\right]=-\E \left[\frac{1}{f_Y(Y)}\frac{d }{d \rho}f_Y(Y) \right],
\end{equation*}
where $(a)$ will be justified in Appendix~\ref{Uniform-Integrability}. Now, some straightforward computations yield
\begin{align*}
\frac{d}{d\rho} f_{Y}(Y) &=\dfrac{d}{d\rho} \int_{\R}  f_{Y|X}(Y|x) f_X(x) dx\\
&\stackrel{(b)}{=}  \int_{\R} f_X(x) \dfrac{d}{d\rho} f_{Y|X}(Y|x) dx\\
&= -\int_{\R}  (\rho X+Z-\rho x)(X-x) f_{Y|X}(Y|x) f_X(x) dx\\
&= -f_Y(Y) \int_{\R} (\rho X+Z-\rho x) (X-x) f_{X|Y}(x|Y) dx,
\end{align*}
where $(b)$ will be justified in Appendix~\ref{Uniform-Integrability}. It then follows that
\begin{align*}
\dfrac{d}{d \rho} I(X; Y)  & = \E \left[\int_{\R} (Y-\rho x) (X-x) f_{X|Y}(x|Y) dx \right] \\
&= \E[YX-Y\E[X|Y]-\rho X \E[X|Y]+\rho \E[X^2|Y]]\\
&= \E[YX]-\E[YX]-\E[\rho \E^2[X|Y]]+\E[\rho \E[X^2|Y]]\\
&= \rho \E[X^2-\E^2[X|Y]]\\
&\stackrel{(c)}{=} \rho \E[(X-\E[X|Y])^2],
\end{align*}
where $(c)$ is due to the orthogonality principle.

\subsection{A new proof of de Bruijn's identity.} \label{new-proof-2}

The following de Bruijn's identity is a fundamental relationship between the differential entropy and the Fisher information~\cite{co2006}. Based on the proposed approach, we will give a new proof of this classical result.

\begin{thm}
Let $X$ be any random variable with a finite variance and let $Z$ be an independent standard normally distributed random variable. Then, for any $t > 0$,
\begin{equation} \label{entropy-fisher-information}
\frac{d}{dt} H(X+\sqrt{t}Z)=\frac{1}{2} J(X+\sqrt{t}Z),
\end{equation}
where $J(\cdot)$ denotes the Fisher information.
\end{thm}

\begin{proof}
First of all, define
\begin{equation*}
Y=X+\sqrt{t} Z,
\end{equation*}
whose density function can be computed as
\begin{equation*}
f_Y(y)=\int_{\mathbb{R}} f_X(x) f_{Y|X}(y|x)dx = \int_{\mathbb{R}} \frac{f_X(x)}{\sqrt{2\pi t}} e^{-(y-x)^2/(2t)} dx.
\end{equation*}
Here, to prevent possible confusion, we remark that $Y$ defined as above should be regarded as ``local'' to this proof, as the same notation is used
to denote the output of Gaussian channels elsewhere in this paper. Immediately, we have
\begin{equation*}
f_Y(Y)=f_Y(X+\sqrt{t}Z)=\int_{\mathbb{R}} \frac{f_X(x)}{\sqrt{2\pi t}} e^{-(X+\sqrt{t}Z-x)^2/(2t)} dx.
\end{equation*}
Now, taking the derivative with respect to $t$, we obtain
\begin{align*}
\frac{d}{dt}  f_Y(Y) & \stackrel{(a)}{=} \int_{\mathbb{R}} \frac{f_X(x)}{\sqrt{2\pi t}} e^{-(X+\sqrt{t}Z-x)^2/(2t)} \left(\frac{(X-x)(X+\sqrt{t}Z-x)}{2t^2}-\frac{1}{2 t} \right)dx\\
&= \int_{\mathbb{R}} \left(\frac{(X-x)(X+\sqrt{t}Z-x)}{2t^2}-\frac{1}{2 t} \right) f_{Y|X}(Y|x) f_X(x) dx\\
&=f_Y(Y) \int_{\mathbb{R}} \left( \frac{(X-x)(Y-x)}{2 t^2}-\frac{1}{2t} \right) f_{X|Y}(x|Y) dx,
\end{align*}
where $(a)$ will be justified in Appendix~\ref{deBruijn}. It then follows that
\begin{align*}
\frac{d}{dt}  H(Y) &= -\frac{d}{dt} \E\left[\log f_Y(Y) \right]\\
&\stackrel{(b)}{=}-\E\left[\frac{1}{f_Y(Y)}\frac{d}{dt}f_Y(Y)\right] \\
&=\E\left[\int_{\mathbb{R}} \left(-\frac{(X-x)(Y-x)}{2t^2}+\frac{1}{2t} \right) f_{X|Y}(x|Y) dx \right]\\
&=\frac{\E[-XY+(X+Y) \E[X|Y]-\E[X^2|Y]]}{2t^2}+\frac{1}{2t}\\
&=\frac{-\E[X^2]+\E[\E^2[X|Y]]}{2t^2}+\frac{1}{2t}\\
&\stackrel{(c)}{=}\frac{-\E[X^2]+\E[\E^2[X|Y]]+\E[(X-Y)^2]}{2t^2}\\
&=\frac{\E[\E^2[X|Y]]+\E[Y^2]-2\E[XY]}{2 t^2},
\end{align*}
where $(b)$ will be justified in Appendix~\ref{deBruijn} and we have used the fact that $t=\E[(X-Y)^2]$ in $(c)$. On the other hand, similarly as above, we derive
\begin{equation*}
f'_Y(Y)=\int_{\mathbb{R}} \frac{f_X(x)}{\sqrt{2 \pi t}} e^{-(Y-x)^2/(2t)} \frac{x-Y}{t} dx = f_Y(Y) \int_{\mathbb{R}} \frac{x-Y}{t} f_{X|Y}(x|Y)  dx,
\end{equation*}
where $f'_Y(\cdot)$ means the derivative of the function of $f_Y(\cdot)$ with respect to its parameter. It then follows that
\begin{align*}
J(Y)&=\E\left[\left(\frac{f_Y'(Y)}{f_Y(Y)}\right)^2 \right]\\
&= \frac{\E[\E^2[X|Y]+Y^2-2\E[X|Y]Y]}{t^2}\\
&=\frac{\E[\E^2[X|Y]]+\E[Y^2]-2\E[XY]}{t^2},
\end{align*}
which immediately implies the desired (\ref{entropy-fisher-information}).
\end{proof}

\begin{rem}
The new proof of de Bruijn's identity actually further reveals that
\begin{equation*}
\frac{d}{dt} H(X+\sqrt{t}Z)=\frac{1}{2} J(X+\sqrt{t}Z)=\frac{1}{2 t^2} \E[(Y-\E[X|X+\sqrt{t}Z])^2].
\end{equation*}
\end{rem}

\section{The Extended I-MMSE Relationship in Discrete Time}  \label{main-results-1}

In this section, using the ideas and techniques illustrated in Section~\ref{new-proofs}, we give extensions of the I-MMSE relationship (\ref{I-MMSE}) to channels with feedback and/or memory.

We start with the following general theorem on a discrete-time system:
\begin{thm} \label{discrete-time-main-result}
Consider the following discrete-time system
\begin{equation} \label{discrete-time-system}
Y_i={\rho}g_i(W_1^i, Y_1^{i-1})+Z_i,\quad i=1, \ldots, n,
\end{equation}
where $\rho \in \mathbb{R}$, all $W_i$ are independent of all $Z_i$, which are i.i.d. standard normal random variables and each $g_i(\cdot, \cdot)$ is a deterministic function differentiable in its second parameter. Assume that for any $i$ and any compact subset $K \subset \mathbb{R}$,
\begin{equation} \label{discrete-condition-d-1}
\E\left[\sup_{\rho \in K} g_i^2(W_1^i, Y_1^{i-1})\right] < \infty, \quad \E\left[\sup_{\rho \in K} \left(\frac{d}{d \rho}g_i(W_1^i, Y_1^{i-1})\right)^2\right] < \infty,
\end{equation}
\begin{equation} \label{discrete-condition-d-2}
\E\left[\sup_{\rho \in K} \E_W^2\left[\left.\frac{d}{d\rho} g_i(W_1^i, Y_1^{i-1})\right|Y_0^T\right] \right] < \infty,
\end{equation}
where
$$
\E_w\left[\left.\frac{d}{d\rho} g_i(W_1^i, Y_1^{i-1})\right|Y_0^T\right] \triangleq \E \left[\left.\frac{d}{d\rho} g_i(w_1^i, Y_1^{i-1})\right|Y_0^T\right].
$$
Then we have
\begin{align} \label{to-be-symmetrized}
\frac{d}{d\rho} I(W_1^n; Y_1^n) & =\rho \sum_{i=1}^n \E\left[(g_i-\E[g_i|Y_1^n])^2\right]+\rho^2 \sum_{i=1}^n \E\left[g_i \E_W \left[\left. \frac{d}{d\rho} g_i \right|Y_1^n \right]-\E[g_i|Y_1^n] \frac{d}{d\rho} g_i\right] \notag\\
&\hspace{3cm}+ \rho \sum_{i=1}^n \E\left[Y_i \left(\frac{d}{d\rho} g_i-\E_W \left[\left. \frac{d}{d\rho} g_i \right|Y_1^n \right] \right) \right],
\end{align}
where we have simply written $g_i(W_1^i, Y_1^{i-1})$ simply as $g_i$, $\E_W\left[\left. \displaystyle{\frac{d}{d\rho} g_i}(W_1^i, Y_1^{i-1})\right|Y_0^T\right]$ simply as $\E_W\left[\left.\displaystyle{\frac{d}{d\rho} g_i} \right|Y_0^T\right]$.
\end{thm}

\begin{proof}
Note that
\begin{equation*}
I(W_1^n; Y_1^n) = H(Y_1^n)-\sum_{i=1}^n H(Y_i|W_1^n, Y_1^{i-1}) = H(Y_1^n)- nH(Z_1),
\end{equation*}
which immediately implies
\begin{equation}  \label{Shamai-0}
\frac{d}{d\rho} I(W_1^n; Y_1^n) =-\frac{d}{d\rho} \E\left[\log f_{Y_1^n}(Y_1^n)\right] \stackrel{(a)}{=} -\E\left[\frac{d}{d\rho} \log f_{Y_1^n}(Y_1^n)\right]=-\E\left[\frac{1}{f_{Y_1^n}(Y_1^n)}\frac{d}{d\rho}f_{Y_1^n}(Y_1^n)\right],
\end{equation}
where $(a)$ will be justified in Appendix~\ref{discrete-interchange}.

In the remainder of the proof, we will omit the subscripts of the density functions. For instance, $f(y_1^n)$ means the density function of $Y_1^n$, $f(Y_1^n)$ means the density function of $Y_1^n$ evaluated at $Y_1^n$, $f(y_1^n|w_1^n)$ means the conditional density function of $Y_1^n$ given $W_1^n=w_1^n$.

Under the system assumptions, we have
\begin{equation*}
f(y_1^n|w_1^n) = \prod_{i=1}^n f(y_i|y_1^{i-1}, w_1^n)=\frac{1}{(\sqrt{2\pi})^n} \prod_{i=1}^n \exp\{-(y_i-\rho g_i(w_1^i, y_1^{i-1}))^2/2\},
\end{equation*}
and furthermore,
\begin{align*}
\frac{d}{d\rho} f(Y_1^n|w_1^n)
&= \frac{1}{(\sqrt{2\pi})^n} \frac{d}{d\rho} \prod_{i=1}^n \exp\{-(Y_i-\rho g_i(w_i, Y_1^{i-1}))^2/2\} \\
&= \frac{1}{(\sqrt{2\pi})^n}  \frac{d}{d\rho} \prod_{i=1}^n \exp\{-(\rho g_i(W_1^i, Y_1^{i-1})-\rho g_i(w_1^i, Y_1^{i-1})+Z_i)^2/2\} \\
&= -f(Y_1^n|w_1^n)\sum_{i=1}^n (Y_i-\rho g_i(w_1^i, Y_1^{i-1})) \bigg(g_i(W_1^i, Y_1^{i-1})-g_i(w_1^i, Y_1^{i-1})
\\
& \quad \quad +\rho \frac{d}{d\rho} (g_i(W_1^i, Y_1^{i-1})- g_i(w_1^i, Y_1^{i-1}))\bigg).
\end{align*}
It then follows that
\begin{align*}
\frac{d}{d\rho}f(Y_1^n) &= \frac{d}{d\rho}\int_{\R^n}  f(Y_1^n|w_1^n) f(w_1^n)dw_1^n\\
&\stackrel{(b)}{=} \int_{\R^n} \frac{d}{d\rho} f(Y_1^n|w_1^n) f(w_1^n)dw_1^n\\
&= -\int_{\R^n} \sum_{i=1}^n (Y_i-\rho g_i(w_1^i, Y_1^{i-1})) \bigg(g_i(W_1^i, Y_1^{i-1})-g_i(w_1^i, Y_1^{i-1})\\
&\quad\quad \quad \quad\quad\quad  +\rho \frac{d}{d\rho} (g_i(W_1^i, Y_1^{i-1})- g_i(w_1^i, Y_1^{i-1}))\bigg) f(Y_1^n|w_1^n) f(w_1^n)dw_1^n\\
&= -f(Y_1^n)\int_{\R^n}\sum_{i=1}^n (Y_i-\rho g_i(w_1^i, Y_1^{i-1}) \bigg( g_i(W_1^i, Y_1^{i-1})-g_i(w_1^i, Y_1^{i-1})\\
&\quad\quad \quad \quad\quad\quad  +\rho \frac{d}{d\rho} (g_i(W_1^i, Y_1^{i-1})- g_i(w_1^i, Y_1^{i-1}))\bigg) f(w_1^n|Y_1^n) dw_1^n,
\end{align*}
where $(b)$ will be justified in Appendix~\ref{discrete-interchange}. Writing $g_i(W_1^i,Y_1^{i-1}), g_i(w_1^i, Y_1^{i-1})$ as $g_i, \tilde{g}_i$ respectively, we have
\begin{align}
\frac{d}{d\rho}f(Y_1^n) & =-f(Y_1^n)\sum_{i=1}^n \int_{\R^n} (Y_i-\rho \tilde g_i)\left((g_i+\rho\frac{d}{d\rho}g_i)-(\tilde g_i+ \rho\frac{d}{d\rho}\tilde g_i)\right) f(w_1^n|Y_1^n)dw_1^n \nonumber \\
& \hspace{-3cm} =-f(Y_1^n)\sum_{i=1}^n \left((g_i+\rho\frac{d}{d\rho}g_i) (Y_i-\rho \E[g_i|Y_1^n])-\E\left[ \left.(Y_i- \rho g_i) g_i \right|Y_1^n\right]-\rho \int_{\mathbb{R}^n} f(w_1^n|Y_1^n) (Y_i-\rho \tilde{g}_i) \frac{d}{d\rho}\tilde{g}_i  dw_1^n \right), \label{Shamai-1}
\end{align}
where we have used the fact that for any measurable function $\varphi$,
\begin{equation*}
\int_{\R^n}\varphi(w_1^n, Y_1^{n}) f(w_1^n|Y_1^n)dw_1^n=\E[\varphi(W_1^n, Y_1^n)|Y_1^n].
\end{equation*}
Using the shorthand nation (below note that for fixed $w_1^i$, $g_i(w_1^i, Y_1^{i-1})$ is a deterministic function of $Y_1^{i-1}$, but $\frac{d}{d\rho} g_i(w_1^i, Y_1^{i-1})$ may \textbf{not} be a deterministic function of $Y_1^{i-1}$)
\begin{equation}  \label{bracket-drho}
\left[\frac{d}{d\rho} g_i \right](w_1^i, y_1^{i-1}) \triangleq \E\left[\left. \frac{d}{d\rho} g_i(w_1^i, Y_1^{i-1}) \right|Y_1^n=y_1^n \right],
\end{equation}
and plugging (\ref{Shamai-1}) into (\ref{Shamai-0}), we continue as in the proof of (\ref{rho-I-MMSE}) to obtain
{\small \begin{align*}
\hspace{-2cm} \frac{d}{d\rho} I(W_1^n; Y_1^n)
&= \sum_{i=1}^n \left(\E \left[(g_i+\rho\frac{d}{d\rho}g_i)(Y_i-\rho\E\left[ \left.  g_i \right|Y_1^n\right])\right]-\E\left[ g_i (Y_i-\rho g_i)\right]-\rho \E\left[\int_{\mathbb{R}^n} f(w_1^n|Y_1^n) (Y_i-\rho \tilde{g}_i) \E\left[\left.\frac{d}{d\rho} \tilde{g}_i\right|Y_1^n \right] dw_1^n \right]\right)\\
&= \sum_{i=1}^n \left(\E \left[(g_i+\rho\frac{d}{d\rho}g_i)(Y_i-\rho\E\left[ \left.  g_i \right|Y_1^n\right])\right]-\E\left[ g_i (Y_i-\rho g_i)\right]-\rho \E\left[\E\left[ \left.(Y_i-\rho g_i) \left[\frac{d}{d\rho} g_i\right]\right|Y_1^n \right]\right] \right)\\
&= \sum_{i=1}^n \left(\E \left[(g_i+\rho\frac{d}{d\rho}g_i)(Y_i-\rho\E\left[ \left.  g_i \right|Y_1^n\right])\right]-\E\left[ g_i (Y_i-\rho g_i)\right]-\rho \E\left[(Y_i-\rho g_i) \left[\frac{d}{d\rho} g_i\right]\right] \right)\\
&=\sum_{i=1}^n \left(\rho \E[g_i^2-g_i \E\left[g_i|Y_1^n] \right] + \rho \E\left[Y_i \left(\frac{d}{d\rho} g_i-\left[\frac{d}{d\rho} g_i\right] \right) \right]+\rho^2 \E\left[g_i \left[\frac{d}{d\rho} g_i \right]-\E[g_i|Y_1^n] \frac{d}{d\rho} g_i\right]\right)\\
&=\sum_{i=1}^n \left(\rho \E[g_i^2-\E^2\left[g_i|Y_1^n] \right] + \rho \E\left[Y_i \left(\frac{d}{d\rho} g_i-\left[\frac{d}{d\rho} g_i\right] \right) \right]+\rho^2 \E\left[g_i \left[\frac{d}{d\rho} g_i \right]-\E[g_i|Y_1^n] \frac{d}{d\rho} g_i\right]\right)\\
&=\sum_{i=1}^n \left(\rho \E[(g_i-\E\left[g_i|Y_1^n])^2 \right] + \rho \E\left[Y_i \left(\frac{d}{d\rho} g_i-\left[\frac{d}{d\rho} g_i\right] \right) \right]+\rho^2 \E\left[g_i \left[\frac{d}{d\rho} g_i \right]-\E[g_i|Y_1^n] \frac{d}{d\rho} g_i\right]\right)\\
&\stackrel{(c)}{=}\rho \sum_{i=1}^n \E[(g_i-\E\left[g_i|Y_1^n])^2 \right] + \rho^2 \sum_{i=1}^n \E\left[g_i \left[\frac{d}{d\rho} g_i \right]-\E[g_i|Y_1^n] \frac{d}{d\rho} g_i\right]+ \rho \sum_{i=1}^n \E\left[Y_i \left(\frac{d}{d\rho} g_i-\left[\frac{d}{d\rho} g_i\right] \right) \right],
\end{align*}}
where $(c)$ is due to the orthogonality principle.
\end{proof}

\begin{rem}
A subtle point is that in general
$$
\E_W\left[\left.\frac{d}{d\rho} g_i(W_1^i, Y_1^{i-1})\right|Y_0^T\right] \neq \E\left[\left.\frac{d}{d\rho} g_i(W_1^i, Y_1^{i-1})\right|Y_0^T\right],
$$
since, by definition, we have
$$
\E_W\left[\left.\frac{d}{d\rho} g_i(W_1^i, Y_1^{i-1})\right|Y_0^T\right]= \left[\frac{d}{d\rho} g_i \right](W_1^i, Y_1^{i-1}),
$$
where $\left[\displaystyle{\frac{d}{d\rho} g_i} \right]$ is defined in (\ref{bracket-drho}).
\end{rem}

\begin{rem}
For each $i$, $g_i(W_1^i, Y_1^{i-1})$ may depend on $\rho$ through its second parameter $Y_1^{i-1}$, which obviously depends on $\rho$. Theorem~\ref{discrete-time-main-result} still holds true even if the function $g_i(\cdot, \cdot)$ itself is parameterized by $\rho$, which, with $\rho^2$ interpreted as SNR, can be of use in applications involving power adjusting schemes. Note that when $g_i(W_1^i, Y_1^{i-1})$ does not depend on $\rho$, the second inequality in (\ref{discrete-condition-d-1}) and (\ref{discrete-condition-d-2}) are vacuously true, and the first inequality boils down to the usual average power constraint.

It is very conceivable that Conditions (\ref{discrete-condition-d-1}) and (\ref{discrete-condition-d-2}) will hold true for $W_1^i$ with ``commonly used'' distribution and most ``practical'' $g_i$; in particular, it is true when each $W_1^i$ is Gaussian distributed and each $g_i$ is a linear function of its parameters.
\end{rem}

\begin{rem} \label{extension-feedback}
Consider the discrete-time system as in (\ref{discrete-time-system}). Rewriting all $W_i$ as $M$ and each $g_i$ as $X_i$, we then have the following discrete-time Gaussian channel with feedback:
\begin{equation} \label{interpretation-1}
Y_i=\sqrt{snr} X_i(M, Y_1^{i-1})+Z_i,\quad i=1, 2, \ldots, n
\end{equation}
where $M$ is interpreted as the message to be transmitted and $X_i, Y_i$ are the channel inputs, outputs, respectively. It is well known that for such a feedback channel,
\begin{equation*}
I(X_1^n \rightarrow Y_1^n)=I(M; Y_1^n),
\end{equation*}
where $I(X_1^n \rightarrow Y_1^n)$ is the directed information~\cite{massey90} between $X_1^n$ and $Y_1^n$. Then, applying Theorem~\ref{discrete-time-main-result} and the chain rule for taking derivative
\begin{equation*}
\frac{d}{d\rho}=\frac{1}{2 \sqrt{snr}} \frac{d}{dsnr}
\end{equation*}
twice, we have
\begin{align} \label{extended-formula-feedback}
\hspace{-1cm} \frac{d}{d snr} I(X_1^n \rightarrow Y_1^n) & = \frac{1}{2} \sum_{i=1}^n \E[(X_i-\E\left[X_i|Y_1^n])^2 \right] + snr \sum_{i=1}^n \E\left[X_i \E_W\left[\left.\frac{d}{d snr} X_i \right|Y_1^n \right]-\E[X_i|Y_1^n] \frac{d}{d snr} X_i\right] \notag\\
&\hspace{2cm}+\sqrt{snr} \sum_{i=1}^n \E\left[Y_i \left(\frac{d}{d snr} X_i-\E_W\left[\left.\frac{d}{d snr} X_i \right|Y_1^n\right] \right) \right],
\end{align}
where $X_i=X_i(M, Y_1^{i-1})$. This yields an extension of the I-MMSE relationship to discrete-time Gaussian channels with feedback.
\end{rem}

\begin{rem} \label{extension-memory}
Alternatively, rewriting each $W_i$ as $X_i$, we will have the following discrete-time Gaussian channel with input and output memory (it is observed that such a channel is suitable for modeling some storage systems, such as flash memories~\cite{me14}):
\begin{equation*}
Y_i=\sqrt{snr} g_i(X_1^i, Y_1^{i-1})+Z_i,\quad i=1, 2, \ldots, n
\end{equation*}
where $g_i$ is interpreted as ``part'' of the channel and $X_i, Y_i$ are the channel inputs, outputs, respectively. Then, by Theorem~\ref{discrete-time-main-result} and the chain rule, we obtain
\begin{align} \label{extended-formula-memory}
\hspace{-1cm} \frac{d}{d snr} I(X_1^n; Y_1^n) & = \frac{1}{2} \sum_{i=1}^n \E[(g_i-\E\left[g_i|Y_1^n])^2 \right] + snr \sum_{i=1}^n \E\left[g_i \E_W\left[\left.\frac{d}{d snr} g_i \right|Y_1^n \right]-\E[X_i|Y_1^n] \frac{d}{d snr} g_i\right] \notag\\
&\hspace{2cm}+\sqrt{snr} \sum_{i=1}^n \E\left[Y_i \left(\frac{d}{d snr} g_i-\E_W\left[\left.\frac{d}{d snr} g_i \right|Y_1^n\right] \right) \right]
\end{align}
where $g_i=g_i(X_1^i, Y_1^{i-1})$. This yields an extension of the I-MMSE relationship to discrete-time Gaussian channels with input and output memory.
\end{rem}

\begin{rem}
Consider the Gaussian feedback channel (\ref{interpretation-1}) satisfying the average power constraint: $\sum_{i=1}^n \E[X_i^2]/n \leq 1$ for all $n$. It has been established by Cover and Pombra~\cite{co1989} that the channel inputs taking the following form can achieve the capacity as $n$ tends to infinity:
\begin{equation} \label{c-p-k-1}
X_1^n=V_1^n+ B_n Y_1^n,
\end{equation}
where $V_1^n$, $Y_1^n$ are jointly Gaussian and $B_n$ is a strictly lower triangular matrix. It can be checked that for such a coding scheme, the right hand side of (\ref{extended-formula-feedback}) boils down to 
$$
\frac{1}{2} \sum_{i=1}^n \E[(X_i-\E\left[X_i|Y_1^n])^2 \right] + snr \sum_{i=1}^n \E\left[(X_i -\E[X_i|Y_1^n]) \frac{d}{d snr} V_i\right],
$$ 
Moreover, we note that Theorem $4.1$ in~\cite{ki10} implies that for the purpose of achieving the capacity, one can choose $V_1^n$ such that its power is ``close'' to $0$, which in turn implies that
\begin{equation}  \label{c-p-k-3}
\frac{\sum_{i=1}^n \E\left[\left(X_i-\E[X_i|Y_1^n]\right)\frac{d}{d snr} V_i\right]}{n} \mbox{ is ``close'' to $0$}.
\end{equation}
It then follows that for the channel (\ref{interpretation-1}) operating at a fixed SNR, say, $snr_0$,
{\small \begin{align}
\hspace{-2.2cm} I(M; Y_1^{(snr_0),n}) & = \int_0^{snr_0} \frac{1}{2} \sum_{i=1}^n \E\left[(X_i^{(snr)}-\E[X_i^{(snr)}|Y_1^{(snr),n}])^2\right]+snr \sum_{i=1}^n \E\left[\left(X_i^{(snr)}-\E[X_i^{(snr)}|Y_1^{(snr),n}]\right)\frac{d}{d snr} V_i^{(snr)}\right] d snr \\
& \approx \int_0^{snr_0} \frac{1}{2} \sum_{i=1}^n \E\left[(X_i^{(snr)}-\E[X_i^{(snr)}|Y_1^{(snr),n}])^2\right] d snr\\
& \stackrel{(a)}{\leq} \int_0^{snr_0} \frac{1}{2} \sum_{i=1}^n \E\left[(X_i^{(snr)}-\E[X_i^{(snr)}|Y_i^{(snr)}])^2\right] d snr \\
& \stackrel{(b)}{\leq} \frac{1}{2} \log (1+ snr_0),
\end{align}}
where we have used parenthesized superscripts ${}^{(snr_0)}, {}^{(snr)}$ to indicate the underlying parameter, $(a)$ is due to Proposition $11$ in~\cite{gu11}, and $(b)$ is due to Proposition $13$ in~\cite{gu11} and the concavity of the $\log$ function.

Note that when there is no feedback, it is well known that the capacity of the channel (\ref{interpretation-1}) operating at SNR $snr_0$ is $\frac{1}{2}\log(1+snr_0)$. So, making use of the extended I-MMSE relation (\ref{extended-formula-feedback}), we have recovered the well-known fact that feedback does not increase the capacity of a white Gaussian channel. Though it exemplifies a way to ``tame'' the correctional term by restricting attention to appropriately chosen encoding schemes, the above argument is heuristic in nature: there is a major technical gap when applying Theorem $4.1$ in~\cite{ki10}, where the result is stated in a ``mutual information rate'' form, as opposed to the ``$n$-block mutual information'' form in the extended I-MMSE relation (\ref{extended-formula-feedback}). A completely rigorous treatment may require a limiting version of (\ref{extended-formula-feedback}), which has been listed as one of possible future directions; see Section~\ref{further-extensions}.
\end{rem}

\section{The Extended I-MMSE Relationship in Continuous Time}  \label{main-results-2}

As elaborated in the following theorem, the continuous-time I-MMSE relationship, the continuous-time analog of (\ref{I-MMSE}), has been established in~\cite{gu05}.

\begin{thm}[Theorem $6$ of~\cite{gu05}] \label{continuous-time-I-MMSE}
Consider the following continuous-time Gaussian channel
\begin{equation*}
Y(t) = \sqrt{snr} \int_0^t X(s) ds + B(t), \quad t \in [0, T],
\end{equation*}
where $\{X(s)\}$ is the channel input satisfying the power constraint
\begin{equation} \label{square-integrable}
\int_0^T \E[X^2(s)] ds < \infty
\end{equation}
and $\{B(t)\}$ is the standard Brownian motion. Then, we have
\begin{equation} \label{continuous-I-MMSE}
\frac{d}{d snr} I(X_0^T; Y_0^T) =\frac{1}{2} \int_0^T \E[\left(X(s)-\E[X(s)|Y_0^T]\right)^2] ds.
\end{equation}
\end{thm}
\noindent In this section, using the ideas and techniques illustrated in Section~\ref{new-proofs}, we give extensions of the continuous-time I-MMSE relationship to channels with feedback or memory.

We start with a general theorem on a continuous-time system:
\begin{thm} \label{continuous-time-main-result-1}
Consider a continuous-time system characterized by the following stochastic differential equation:
\begin{align}\label{continuous-time-system}
Y(t) = \rho\int_0^t g(s, W_0^s, Y_0^s) ds+B(t),\quad t \in[0, T],
\end{align}
where $\rho \in \mathbb{R}$, the stochastic process $\{W(t)\}$ is independent of the standard Brownian motion $\{B(t)\}$, and $g(\cdot, \cdot, \cdot)$ is a deterministic function. Assume that
\begin{itemize}
\item[(a)] $g(s, \gamma_0^s, \phi_0^s)$ is defined for all $\gamma(\cdot), \phi(\cdot) \in C[0, T]$, the set of all continuous functions over $[0, T]$;
\item[(b)] the solution $\{Y(t)\}$ to the stochastic differential equation (\ref{continuous-time-system}) uniquely exists;
\item[(c)] for any $s \in [0, T]$, $g(s, W_0^s, Y_0^s)$ is differentiable with respect to $\rho$ with probability $1$;
\item[(d)] for any compact subset $K \subset \mathbb{R}$, there exists a constant $\eps_0 > 0$ such that
\begin{equation} \label{d-1}
\int_0^T \E\left[\sup_{\rho \in K} g^2(s, W_0^s, Y_0^s) \right] ds < \infty, \quad \int_0^T \E\left[\sup_{\rho \in K, w_0^T \in C[0, T]} \left(\frac{d}{d \rho}g(s, w_0^s, Y_0^s)\right)^{8+\eps_0}\right] ds < \infty,
\end{equation}
and there exist two constants $C_1, \eps_1 > 0$ such that for any $\rho_1, \rho_2 \in \mathbb{R}$,
\begin{equation} \label{d-2}
\int_0^T \E\left[\sup_{w_0^T \in C[0, T]} \left.\left(\frac{d}{d \rho}g(s, w_0^s, Y_{0}^{s})\right|_{\rho=\rho_1}-\left.\frac{d}{d \rho}g(s, w_0^s, Y_0^{s})\right|_{\rho=\rho_2}\right)^{2+\eps_1}\right] \leq C_1 |\rho_1-\rho_2|^{2+\eps_1}.
\end{equation}
\item[(e)] there exists a constant $C_2 > 0$ such that for all all $\gamma(\cdot), \phi(\cdot) \in C[0, T]$,
\begin{equation*}
\int_0^T g^2(s, \gamma_0^s, \phi_0^s) ds < C_2.
\end{equation*}
\end{itemize}
Then, we have
\begin{align} \label{continuous-time-feedback-formula-1}
\hspace{-1.5cm} \frac{d}{d\rho} I(W_0^T; Y_0^T) &= \rho\int_0^T \E[(g(s)-\E[g(s)|Y_0^T])^2] ds+\rho^2 \int_0^T \E\left[ g(s) \E_W\left[\left.\frac{d}{d\rho} g(s)\right|Y_0^T \right]-\frac{d}{d\rho} g(s) \E[g(s)|Y_0^T]\right] ds \nonumber\\
&\hspace{2cm}+\rho^2\int_0^T \E\left[g(s) \frac{d}{d\rho} g(s) \right]ds-\rho \E\left[\left.\E_W\left[\int_0^T \frac{d}{d\rho} g(s) dY(s)\right|Y_0^T \right]\right],
\end{align}
where we have written $g(s, W_0^s,Y_0^{s})$ simply as $g(s)$ and defined:
$$
\E_w\left[\left.\frac{d}{d\rho} g(s, W_0^s, Y_0^s)\right|Y_0^T \right] \triangleq \E\left[\left.\frac{d}{d\rho} g(s, w_0^s, Y_0^s)\right|Y_0^T \right],
$$
$$
\E_w\left[\left. \int_0^T \frac{d}{d\rho} g(s, W_0^s, Y_0^s) dY(s)\right|Y_0^T \right] = \E \left[\left. \int_0^T \frac{d}{d\rho} g(s, w_0^s, Y_0^s) dY(s)\right|Y_0^T \right].
$$
\end{thm}

Strictly speaking, Theorem~\ref{continuous-time-main-result-1} is not a generalization of Theorem~\ref{continuous-time-I-MMSE}: Condition (e) is stronger than the square integrability condition (\ref{square-integrable}), as one can easily find $g$ satisfying the latter but not the former. Condition (e) will be relaxed in the following theorem, which ``essentially'' generalizes Theorem~\ref{continuous-time-I-MMSE}.

\begin{thm}   \label{continuous-time-main-result-2}
Consider the continuous-time system (\ref{continuous-time-system}) satisfying Conditions (a), (b), (c), (d) and the following conditions:
\begin{enumerate}
\item[(f)] for any constant $a > 0$ and any $t \in [0, T]$,
\begin{equation*}
P\left( \int_0^t g^2(s, W_0^s, Y_0^s) ds = a \right)=0;
\end{equation*}
\item[(g)] with probability $1$, we have (note that the third parameter in the following $g$ function is $B_0^s$, rather than $Y_0^s$)
\begin{equation*}
\int_0^T g^2(s, W_0^s, B_0^s) ds < \infty.
\end{equation*}
\end{enumerate}
Then, we have
\begin{align} \label{continuous-time-feedback-formula-2}
\hspace{-1.5cm} \frac{d}{d\rho} I(W_0^T; Y_0^T) &= \rho\int_0^T \E[(g(s)-\E[g(s)|Y_0^T])^2] ds+\rho^2 \int_0^T \E\left[ g(s) \E_W\left[\left.\frac{d}{d\rho} g(s)\right|Y_0^T \right]-\frac{d}{d\rho} g(s) \E[g(s)|Y_0^T]\right] ds \nonumber\\
&\hspace{2cm}+\rho^2\int_0^T \E\left[g(s) \frac{d}{d\rho} g(s) \right]ds-\rho \E\left[\left.\E_W\left[\int_0^T \frac{d}{d\rho} g(s) dY(s)\right|Y_0^T \right]\right],
\end{align}
where we have written $g(s, W_0^s,Y_0^{s})$ simply as $g(s)$.
\end{thm}

\begin{rem}
Similarly as in discrete time, Theorems~\ref{continuous-time-main-result-1} and~\ref{continuous-time-main-result-2} still hold if for each $s$, the function $g(s, \cdot, \cdot)$ depends on $\rho$.

As ponderous as it may seem, Condition (d) is in fact mild: All the three inequalities can be readily satisfied if $g$ is ``smooth'' enough, and the ``tail part'' of the function $g$ approaches to $0$ ``fast enough'', and so do those of the random elements $W$, $Y$; in particular, when $g$ does not depend on $\rho$, then the second and third inequalities in Condition (d) are vacuously true, and the first inequality boils down to the usual average power constraint.

Despite its deceivingly simple look, Condition (e), reminiscent of the peak power constraint, is somewhat restrictive. At the expense of the extra yet mild Condition (f), Condition (e) is relaxed to a much weaker square integrability Condition (g) in Theorem~\ref{continuous-time-main-result-2}.

Theorem~\ref{continuous-time-main-result-2} does not, however, subsume Theorem~\ref{continuous-time-main-result-1} due to the extra Condition (f), although it is a very weak condition: Condition (f) essentially says that for each $t > 0$, the random variable $\int_0^t g^2(s, W_0^s, Y_0^s) dt$ is ``purely continuous'' without any ``point mass''.
\end{rem}

\begin{rem}
Parallel to Remarks~\ref{extension-feedback}, the continuous-time system in (\ref{continuous-time-system}) can be interpreted as the following continuous-time Gaussian channel with feedback:
\begin{equation} \label{feedback-sampling-conjecture}
Y(t) = \sqrt{snr} \int_0^t X(s, M, Y_0^{s}) ds+B(t),\quad t \in[0, T].
\end{equation}
An application of Theorem~\ref{continuous-time-main-result-1} then yields
{\small \begin{align} \label{continuous-extended-formula-feedback}
\hspace{-2cm} \frac{d}{d snr} I(M; Y_0^T) & =\frac{1}{2} \int_0^T \E[(X(s)-\E[X(s)|Y_0^T])^2] ds+ snr \int_0^T \E\left[X(s) \E_W\left[\left.\frac{d}{d snr} X(s)\right|Y_0^T \right]-\frac{d}{d snr} X(s) \E[X(s)|Y_0^T]\right] ds \nonumber\\
&\hspace{2cm}+snr \int_0^T \E\left[X(s) \frac{d}{d snr} X(s) \right]ds-\sqrt{snr} \E\left[\left.\E_W\left[\int_0^T \frac{d}{d snr} X(s) dY(s)\right|Y_0^T \right]\right],
\end{align}}
where $X(s)$ is the abbreviated form of $X(s, M, Y_0^{s})$. This gives an extension of the I-MMSE relationship to continuous-time Gaussian channels with feedback. It is well shown~\cite{du70, ka71} that
\begin{equation} \label{KZZ}
I(M; Y_0^T)=\frac{snr}{2} \int_0^T \E[(\E[X(s)]-\E[X(s)|Y_0^s])^2] ds.
\end{equation}
This, together with (\ref{continuous-extended-formula-feedback}), gives that for a fixed $snr_0$,
{\small \begin{align*}
&snr_0 \int_0^T \E[(\E[X^{(snr_0)}(s)]-\E[X^{(snr_0)}(s)|Y_0^{(snr_0), s}])^2] ds\\
&=\int_0^{snr_0} \int_0^T \E[(X^{(snr)}(s)-\E[X^{(snr)}(s)|Y_0^{(snr), T}])^2] ds d snr\\
& \hspace{-1cm} + \int_0^{snr_0} 2 snr \int_0^T \E\left[X^{(snr)}(s) \E_W\left[\left.\frac{d}{d snr} X^{(snr)}(s)\right|Y_0^{(snr), T} \right]-\frac{d}{d snr} X^{(snr)}(s) \E[X^{(snr)}(s)|Y_0^{(snr), T}]\right] ds d snr \nonumber\\
&\hspace{-2cm} +\int_0^{snr_0} 2 snr \int_0^T \E\left[X^{(snr)}(s) \frac{d}{d snr} X^{(snr)}(s) \right]ds d snr- \int_0^{snr_0} 2\sqrt{snr} \E\left[\left.\E_W\left[\int_0^T \frac{d}{d snr} X^{(snr)}(s) dY^{(snr)}(s)\right|Y_0^{(snr), T} \right]\right] d snr,
\end{align*}}
which extends the relationship between the causal MMSE and non-causal MMSE obtained in Theorem $8$ of~\cite{gu05} to Gaussian feedback channels.

Parallel to Remark~\ref{extension-memory}, the continuous-time system in (\ref{continuous-time-system}) can also be interpreted as the following continuous-time Gaussian channel with input and output memory:
\begin{equation*}
Y(t) = \sqrt{snr} \int_0^t g(s, X_0^s, Y_0^{s}) ds+B(t),\quad t \in[0, T].
\end{equation*}
An application of Theorem~\ref{continuous-time-main-result-1} then yields
{\small \begin{align} \label{continuous-extended-formula-memory}
\hspace{-2cm} \frac{d}{d snr} I(X_0^T; Y_0^T) & =\frac{1}{2} \int_0^T \E[(g(s)-\E[g(s)|Y_0^T])^2] ds+ snr \int_0^T \E\left[g(s) \E_W\left[\left.\frac{d}{d snr} g(s)\right|Y_0^T \right]-\frac{d}{d snr} g(s) \E[g(s)|Y_0^T]\right] ds \nonumber\\
&\hspace{2cm}+snr \int_0^T \E\left[g(s) \frac{d}{d snr} g(s) \right]ds-\sqrt{snr} \E\left[\left.\E_W\left[\int_0^T \frac{d}{d snr} g(s) dY(s)\right|Y_0^T \right]\right],
\end{align}}
where $g(s)$ is the abbreviated form of $g(s, X_0^s, Y_0^{s})$. This gives an extension of the I-MMSE relationship to continuous-time Gaussian channels with input and output memory. And similarly as before, (\ref{KZZ}), together with (\ref{continuous-extended-formula-memory}), gives that for a fixed $snr_0$,
{\small \begin{align*}
&snr_0 \int_0^T \E[(\E[g^{(snr_0)}(s)]-\E[g^{(snr_0)}(s)|Y_0^{(snr_0), s}])^2] ds\\
&=\int_0^{snr_0} \int_0^T \E[(g^{(snr)}(s)-\E[g^{(snr)}(s)|Y_0^{(snr), T}])^2] ds d snr\\
& \hspace{-1cm} + \int_0^{snr_0} 2 snr \int_0^T \E\left[g^{(snr)}(s) \E_W\left[\left.\frac{d}{d snr} g^{(snr)}(s)\right|Y_0^{(snr), T} \right]-\frac{d}{d snr} g^{(snr)}(s) \E[g^{(snr)}(s)|Y_0^{(snr), T}]\right] ds d snr \nonumber\\
&\hspace{-2cm} +\int_0^{snr_0} 2 snr \int_0^T \E\left[g^{(snr)}(s) \frac{d}{d snr} g^{(snr)}(s) \right]ds d snr- \int_0^{snr_0} 2\sqrt{snr} \E\left[\left.\E_W\left[\int_0^T \frac{d}{d snr} g^{(snr)}(s) dY^{(snr)}(s)\right|Y_0^{(snr), T} \right]\right] d snr,
\end{align*}}
which extends the relationship between the causal MMSE and non-causal MMSE obtained in Theorem $8$ of~\cite{gu05} to Gaussian memory channels.
\end{rem}

\begin{rem}
It can be readily verified that Theorem~\ref{continuous-time-main-result-2}, when interpreted as (\ref{continuous-extended-formula-feedback}) or (\ref{continuous-extended-formula-memory}) in the previous remark, includes Theorem~\ref{continuous-time-I-MMSE} as a special case; see more detailed explanations in Remark~\ref{rigorously-recover}.
\end{rem}

\subsection{Properties of the solution to (\ref{continuous-time-system})}

In this section, we will give certain sufficient conditions that will guarantee the solution $Y$ to (\ref{continuous-time-system}) uniquely exists (Condition (b) in Theorem~\ref{continuous-time-main-result-1}), and moreover, $g(s, W_0^s, Y_0^s)$ is differentiable with respect to $\rho$ (Condition (c) in Theorem~\ref{continuous-time-main-result-1}). More precisely, we have the following proposition.
\begin{pr}
Under the following conditions:
\begin{itemize}
\item $D g(s, \gamma_0^s, \phi_0^s)$, the Frechet derivative of $g$ with respect to its third parameter $\phi(\cdot)$, exists for any $s \in [0, T]$ and any $\gamma(\cdot), \phi(\cdot) \in C[0, T]$;
\item (extended uniform Lipschitz conditions) There exists a constant $C$ such that for all $s \in [0, T]$ and all $\gamma(\cdot), \phi(\cdot), \psi(\cdot) \in C[0, T]$, we have
    \begin{equation} \label{c-1}
    |g(s, \gamma_0^s, \phi_0^s)-g(s, \gamma_0^s, \psi_0^s)| \leq C \|\phi_0^s-\psi_0^s\|_{\infty},
    \end{equation}
    and
    \begin{equation} \label{c-2}
    \|D g(s, \gamma_0^s, \phi_0^s)-D g(s, \gamma_0^s, \psi_0^s)\| \leq C \|\phi_0^s-\psi_0^s\|_{\infty};
    \end{equation}
\item (extended linear growth conditions) There exists a constant $C$ such that for all $s \in [0, T]$ and all $\gamma(\cdot), \phi(\cdot) \in C[0, T]$, we have
    \begin{equation} \label{c-3}
    g^2(s, \gamma_0^s, \phi_0^s) \leq C(1+ \|\gamma_0^s\|_{\infty}^2 +\|\phi_0^s\|_{\infty}^2),
    \end{equation}
    and
    \begin{equation} \label{c-4}
    \|D g(s, \gamma_0^s, \phi_0^s)\|^2 \leq C(1+ \|\gamma_0^s\|_{\infty}^2 +\|\phi_0^s\|_{\infty}^2),
    \end{equation}
\end{itemize}
the solution $Y$ to the continuous-time system (\ref{continuous-time-system}) uniquely exists, and moreover, with probability $1$, $g(s, W_0^s, Y_0^s)$ is differentiable with respect to $\rho$.
\end{pr}

\begin{proof}
We only sketch the proof, as it is essentially the standard argument for the existence and uniqueness of the solution to a stochastic differential equation with the well-known uniform Lipschitz and linear growth conditions; see, e.g., the proof of Theorem $2.2$ in Chapter $5$ of~\cite{mao08}.

Consider the following Picard's iteration:
\begin{equation*}
Y_{(0)}(t) \equiv 0, \quad Y_{(n+1)}(t) = \int_0^t g(s, W_0^s, Y_{(n),0}^s) ds + B(t), \quad t \in[0, T].
\end{equation*}
It can be easily verified that, for any $n$ and any $t \in [0, T]$, $Y_{(n)}(t)$ is differentiable with respect to $\rho$. Letting $Z_{(n)}(t)=\frac{d}{d\rho} Y_{(n)}(t)$ for all $n$, we have
\begin{equation*}
Z_{(0)}(t) \equiv 0, \quad Z_{(n+1)}(t) = \int_0^t g(s, W_0^s, Y_{(n),0}^s) ds+\rho \int_0^t D g(s, W_0^s, Y_{(n), 0}^s)(Z_{(n),0}^s) ds, \quad t \in[0, T].
\end{equation*}
Now, applying the standard argument for the existence and uniqueness of the solution to a stochastic differential equation, we deduce that there exists a stochastic process $\{Y(t), t \in [0, T]\}$ such that for any compact set $K \subset \mathbb{R}$,
\begin{equation*}
\lim_{n \to \infty} \sup_{\rho \in K, \, t \in [0, T]} |Y_n(t)-Y(t)|=0, \quad \mbox{a.s.}
\end{equation*}
and furthermore, there exists a stochastic process ${Z(t), t \in [0, T]}$ such that for any compact set $K \subset \mathbb{R}$,
\begin{equation*}
\lim_{n \to \infty} \sup_{\rho \in K, \, t \in [0, T]} |Z_n(t)-Z(t)|=0, \quad \mbox{a.s.}
\end{equation*}
It then follows that $Y(t)$ is differentiable with respect to $\rho$, and $\frac{d}{d\rho} Y(t)=Z(t)$ with probability $1$, and consequently, $g(s, W_0^s, Y_0^s)$ is differentiable with respect to $\rho$.
\end{proof}

\begin{rem}
It is well known~\cite{mao08} that (\ref{c-1}) and (\ref{c-3}) are respectively the usual Lipschitz and linear growth conditions that guarantee the existence and uniqueness of the solution to (\ref{continuous-time-system}). The addition of (\ref{c-2}) and (\ref{c-4}) further ensures the differentiability of the solution with respect to $\rho$.
\end{rem}

\subsection{Girsanov's Theorem}
One of the major tools that will be used in our treatment of continuous-time Gaussian channels is Girsanov's theorem. To be more precise, we will use the following two versions of Girsanov's theorem that can be found in~\cite{LS}, which are stated below with slightly different notation from~\cite{LS}.

Let $(\Omega, \mathcal{F}, P)$ be a complete probability space, let $\{\mathcal{F}_t\}, t \geq 0$, be a nondecreasing family of sub-$\sigma$-algebras, and let $B=\{B(t), \mathcal{F}_t\}$, $t \geq 0$, be a standard Brownian motion. For $T > 0$, we consider the It\^{o} process $\xi=(\xi(t), \mathcal{F}_t)$, $0 \leq t \leq T$, such that
\begin{equation}  \label{Ito-1}
\xi(t)=\int_0^t \beta(s) ds + B(t), \quad \xi(0)=0.
\end{equation}
Denote by $(C_T, \mathcal{B}_T)$ the measurable space of the continuous functions $x(s)$, $0 \leq s \leq T$, with $x(0)=0$, and let $\mu_{\xi}$ and $\mu_{B}$ be the measures on $(C_T, \mathcal{B}_T)$ induced by $\xi$ and $B$, respectively.

\begin{thm}[Theorem $7.1$ of~\cite{LS}]  \label{Theorem7.1}
Let $\xi=\{\xi(t), \mathcal{F}_t\}$, $0 \leq t \leq T$, be an It\^{o} process satisfying (\ref{Ito-1}). Suppose that the process $\beta=\{\beta(t), \mathcal{F}_t\}$, $0 \leq t \leq T$, satisfies
\begin{equation*}
P\left(\int_0^T \beta^2(t) dt < \infty\right)=1,
\end{equation*}
and
\begin{equation} \label{expo-integrability}
\E\left[\exp\left\{-\int_0^T \beta(t) dB(t)-\frac{1}{2} \int_0^T \beta^2(t) dt\right\}\right]=1.
\end{equation}
Then, $\mu_{\xi} \sim \mu_{B}$ (here $\sim$ means ``equivalent'') and with probability $1$
\begin{equation*}
\frac{d\mu_B}{d\mu_{\xi}}(\xi)=\E\left[\left.\exp\left\{-\int_0^T \beta(t) d\xi(t)+\frac{1}{2} \int_0^T \beta^2(t) dt\right\}\right|\xi_0^T\right].
\end{equation*}
\end{thm}

Condition (\ref{expo-integrability}) in Theorem~\ref{Theorem7.1} is somewhat restrictive since it is an exponential integrability condition. As elaborated in the following version of Girsanov's theorem, when $\xi$ is a diffusion process (a special kind of It\^{o} process), this condition can be replaced by some much weaker square integrability conditions.
\begin{thm}[Theorem $7.7$ of~\cite{LS}]  \label{Theorem7.7}
Let $\xi=\{\xi(t), \mathcal{F}_t\}$, $0 \leq t \leq T$, be an It\^{o} process satisfying
\begin{equation}  \label{Ito-2}
\xi(t)=\int_0^t \alpha(s, \xi_0^s) ds + B(t), \quad \xi(0)=0,
\end{equation}
where for each $s$, $\alpha(s, \xi_0^s)$ is measurable with respect to $\sigma(\xi_0^s)$, the $\sigma$-algebra generated by $\{\xi_0^s\}$. Then, $\mu_{\xi} \sim \mu_B$ if and only if
\begin{equation} \label{2-squares}
P\left(\int_0^T \alpha^2(t, \xi_0^t) dt < \infty \right)=1, \quad P\left(\int_0^T \alpha^2(t, B_0^t) dt < \infty\right)=1.
\end{equation}
Moreover, if (\ref{2-squares}) holds, we have
\begin{equation*}
\frac{d\mu_{\xi}}{d\mu_{B}}(B)=\exp\left\{\int_0^T \alpha(t, B_0^t) dB(t)-\frac{1}{2} \int_0^T \alpha^2(t, B_0^t) dt\right\},
\end{equation*}
and
\begin{equation*}
\frac{d\mu_B}{d\mu_{\xi}}(\xi)=\exp\left\{-\int_0^T \alpha(t, \xi_0^t) d\xi(t)+\frac{1}{2} \int_0^T \alpha^2(t, \xi_0^t) dt\right\}.
\end{equation*}

\end{thm}

\subsection{Proof of Theorem~\ref{continuous-time-main-result-1}}

Fix $W=w$ and let $Y_{|w)}$ be such that
\begin{equation*}
Y_{|w)}(t)=\rho \int_0^t g(s, w_0^s, Y_{|w), 0}^s) ds+B(t), \quad t \in [0, T].
\end{equation*}
Then, by Theorem~\ref{Theorem7.1} (it can be checked that its assumptions are implied by Condition (e)), we observe that $\mu_{Y_{|w)}}\sim \mu_B \sim \mu_{Y}$, and furthermore,
\begin{equation*}
\frac{d\mu_{Y_{|w)}|W}}{d\mu_B}(Y_{|w), 0}^T|w_0^T)=\exp\left\{\rho\int_0^T g(s, w_0^s, Y_{|w), 0}^s) dY_{|w)}(s)-\frac{\rho^2}{2}\int_0^T g^2(s, w_0^s, Y_{|w), 0}^s)ds\right\}.
\end{equation*}
It then follows from Lemma $4.10$ in~\cite{LS} that
\begin{equation*}
\frac{d\mu_{Y|W}}{d\mu_B}(Y_0^T|w_0^T) =\exp\left\{\rho\int_0^T g(s, w_0^s, Y_0^s) dY(s)-\frac{\rho^2}{2}\int_0^T g^2(s, w_0^s, Y_0^{s}) ds\right\}.
\end{equation*}
Note that, by definition, we have
\begin{align*}
I(W_0^T; Y_0^T)
&=\E\left[\log \dfrac{d\mu_{W Y}}{d(\mu_W \times \mu_{Y})}(W_0^T, Y_0^T)\right]\\
&=\E\left[\log \frac{d\mu_{Y|W}}{d\mu_B}(Y_0^T|W_0^T)\right]-\E\left[\log \frac{d\mu_{Y}}{d\mu_B}(Y_0^T)\right]\\
&=\frac{\rho^2}{2}\int_0^T \E [g^2(s)]ds-\E\left[\log \frac{d\mu_{Y}}{d\mu_B}(Y_0^T)\right].
\end{align*}
Taking derivative with respect to $\rho$ then yields
\begin{align*}
\frac{d}{d\rho }I(W_0^T; Y_0^T)& =\rho\int_0^T \E[g^2(s)]ds+\frac{\rho^2}{2}\frac{d}{d\rho }\int_0^T \E[g^2(s)] ds-\dfrac{d}{d\rho}\E\left[\log \frac{d\mu_{Y}}{d\mu_B}(Y_0^T)\right]\\
&\stackrel{(a)}{=}\rho\int_0^T \E[g^2(s)]ds+\rho^2\int_0^T \E \left[g(s) \frac{d}{d\rho}g(s) \right]ds-\dfrac{d}{d\rho}\E\left[\log \frac{d\mu_{Y}}{d\mu_B}(Y_0^T)\right],
\end{align*}
where $(a)$ will be justified in Appendix~\ref{continuous-interchange}. Writing $g(s, w_0^s,Y_0^s)$ as $\tilde{g}(s)$, we have
\begin{align*}
\frac{d}{d\rho}\left(\frac{d\mu_{Y}}{d\mu_B}(Y_0^T)\right)
&=\frac{d}{d\rho}\int \frac{d\mu_{Y|W}}{d\mu_B}(Y_0^T|w_0^T) \mu_W(dw)\\
&\stackrel{(b)}{=}\int \frac{d}{d\rho} \frac{d\mu_{Y|W}}{d\mu_B}(Y_0^T|w_0^T) \mu_W(dw)\\
&=\int \frac{d}{d\rho} \exp\left\{\rho \int_0^T \tilde{g}(s) dY(s)-\frac{\rho^2}{2}\int_0^T \tilde{g}^2(s) ds\right\} \mu_W(dw)\\
&=\int \frac{d}{d\rho} \exp\left\{\rho^2 \int_0^T \tilde{g}(s) g(s) ds+\rho \int_0^T \tilde{g}(s) dB(s)-\frac{\rho^2}{2}\int_0^T \tilde{g}^2(s) ds\right\} \mu_W(dw)\\
&=\int
 \bigg( \int_0^T \tilde{g}(s) dY(s)+\rho\int_0^T \frac{d}{d\rho} \tilde{g}(s) dY(s)+\rho \int_0^T \tilde{g}(s) (g(s)-\tilde{g}(s))ds\\
&\quad\quad+\rho^2\int_0^T \tilde{g}(s) \frac{d}{d\rho} \left(g(s)-\tilde{g}(s)\right) ds \bigg)\frac{d\mu_{WY}}{d\mu_B}(dw, Y_0^T)
 \\
&=\frac{d\mu_{Y}}{d\mu_B}(Y_0^T)\int  \bigg(\int_0^T \tilde{g}(s) dY(s)+\rho\int_0^T \frac{d}{d\rho} \tilde{g}(s) dY(s)+\rho \int_0^T \tilde{g}(s) (g(s)-\tilde{g}(s))ds\\
&\quad\quad+\rho^2\int_0^T \tilde{g}(s) \frac{d}{d\rho} \left(g(s)-\tilde{g}(s)\right)  ds \bigg)
 \mu_{W|Y}(dw|Y_0^T)\\
&=\frac{d\mu_{Y}}{d\mu_B}(Y_0^T) \bigg(\E \left[\left.\int_0^T g(s) dY(s) \right|Y_0^T \right]
+\rho \int \int_0^T \frac{d}{d\rho} \tilde{g}(s) dY(s) \mu_{W|Y}(dw|Y_0^T)\\
&\quad +\rho\int_0^T (\E[g(s)|Y_0^T]g(s)-\E[g^2(s)|Y_0^T])ds\\
&\quad +\rho^2 \int_0^T \frac{d}{d\rho} g(s) \E[g(s)|Y_0^T]ds - \rho^2 \int_0^T \tilde{g}(s) \frac{d}{d\rho} \tilde{g}(s) ds \mu_{W|Y}(dw|Y_0^T)\bigg),
\end{align*}
where $(b)$ will be justified in Appendix~\ref{continuous-interchange} and we have used $dY(s)=g(s)ds+dB(s)$ for the fifth equality.
Note that by the properties of conditional expectation and the It\^o integral, we have
\begin{equation*}
\E\left[\E\left[\left.\int_0^T g(s) dY(s)\right|Y_0^T \right]\right]=\E\left[\int_0^T g(s) dY(s)\right]=\rho \int_0^T \E [g^2(s)] ds,
\end{equation*}
and similarly,
\begin{equation*}
\E\left[\int_0^T \E[g^2(s)|Y_0^T]ds \right]= \int_0^T \E[g^2(s)]ds,
\end{equation*}
and
\begin{equation*}
\rho\E\left[ \E\left[\left. \int_0^T \frac{d}{d\rho} g(s) dY(s) \right|Y_0^T \right] \right]=\rho^2 \int_0^T \E\left[g(s) \frac{d}{d\rho} g(s) \right]ds=\rho^2\E \left[\int_0^T \E \left[\left.g(s)\frac{d}{d\rho} g(s) \right|Y_0^T\right]ds \right].
\end{equation*}
It then follows that
\begin{align*}
\dfrac{d}{d\rho}\E\left[\log \frac{d\mu_{Y}}{d\mu_B}(Y_0^T)\right] &\stackrel{(c)}{=} \E\left[ \dfrac{d}{d\rho}\log \frac{d\mu_{Y}}{d\mu_B}(Y_0^T)\right]\\
&=\E\left[\dfrac{d}{d\rho}\left(\frac{d\mu_{Y}}{d\mu_B}(Y_0^T)\right)/\frac{d\mu_{Y}}{d\mu_{B}}(Y_0^T)\right]\\
&= \rho\int_0^T \E[g^2(s)]ds+\rho \E\left[\int \E\left[\left.\int_0^T \frac{d}{d\rho} \tilde{g}(s) dY(s)\right|Y_0^T\right] \mu_{W|Y}(dw|Y_0^T)\right]\\
&+\rho \E\left[\int_0^T \E[g(s)|Y_0^T] g(s)-\E[g^2(s)|Y_0^T] ds\right]+\rho^2 \int_0^T \frac{d}{d\rho} g(s) \E[g(s)|Y_0^T] ds\\
&-\rho^2 \E\left[ \int \int_0^T \tilde{g}(s) \E\left[\left. \frac{d}{d\rho} \tilde{g}(s) \right|Y_0^T \right]ds \mu_{W|Y}(dw|Y_0^T)\right]\\
&= \rho \E\left[\int_0^T \E[g(s)|Y_0^T] g(s) ds\right]+\rho \E\left[\E\left[\left.\E_W\left[\left.\int_0^T \frac{d}{d\rho} g(s) dY(s)\right|Y_0^T \right]\right|Y_0^T\right]\right]\\
&+\rho^2 \int_0^T \E\left[\frac{d}{d\rho} g(s) \E[g(s)|Y_0^T]\right] ds-\rho^2 \E\left[\E\left[\left.\int_0^T g(s) \E_W\left[\left.\frac{d}{d\rho} g(s) \right|Y_0^T\right]ds\right|Y_0^T\right]\right]\\
&= \rho \int_0^T \E[\E^2[g(s)|Y_0^T]] ds+\rho \E\left[\left.\E_W\left[\int_0^T \frac{d}{d\rho} g(s) dY(s)\right|Y_0^T \right]\right]\\
&+\rho^2 \int_0^T \E\left[\frac{d}{d\rho} g(s) \E[g(s)|Y_0^T]\right] ds-\rho^2 \E\left[\int_0^T g(s) \E_W\left[\left.\frac{d}{d\rho} g(s)\right|Y_0^T \right]ds\right]
\end{align*}
where $(c)$ will be justified in Appendix~\ref{continuous-interchange}. Straightforward computations then yield that
\begin{align*}
\frac{d}{d\rho }I(W_0^T; Y_0^T) &=\rho\int_0^T \E[(g(s)-\E[g(s)|Y_0^T])^2] ds\\
&+\rho^2 \int_0^T \E\left[ g(s) \E_W\left[\left.\frac{d}{d\rho} g(s)\right|Y_0^T \right]-\frac{d}{d\rho} g(s) \E[g(s)|Y_0^T]\right] ds\\
&+\rho^2\int_0^T \E\left[g(s) \frac{d}{d\rho} g(s) \right]ds-\rho \E\left[\left.\E_W\left[\int_0^T \frac{d}{d\rho} g(s) dY(s)\right|Y_0^T \right]\right],
\end{align*}
as desired.

\subsection{Proof of Theorem~\ref{continuous-time-main-result-2}}

The proof consists of the following $6$ steps:

\noindent {\bf Step $\mathbf{1}$.} First of all, for any fixed $W=w$, by Theorem~\ref{Theorem7.7}, $\mu_{Y|W=w} \sim \mu_B$ with
\begin{equation*}
\frac{d\mu_{Y|W=w}}{d\mu_B}(B_0^T)=\exp \left(\int_0^T g(s, w_0^s, B_0^s) dB(s)-\frac{1}{2} \int_0^T g^2(s, w_0^s, B_0^s) ds \right),
\end{equation*}
where we have used Conditions (d) and (g) before invoking Theorem~\ref{Theorem7.7}. Moreover, by Condition (d), it follows from Theorem $7.2$ of~\cite{LS} that $\mu_Y \ll \mu_B$ with
\begin{align*}
\frac{d\mu_Y}{d\mu_B}(B_0^T) & = \int \frac{d\mu_{Y|W=w}}{d\mu_B}(B_0^T) d\mu_W(w)\\
&= \int \exp \left(\int_0^T g(s, w_0^s, B_0^s) dB(s)-\frac{1}{2} \int_0^T g^2(s, w_0^s, B_0^s) ds \right) d\mu_W(w),
\end{align*}
which is obviously positive with probability $1$. It then follows from Lemma $6.8$ of~\cite{LS} that $\mu_B \ll \mu_Y$. So, in this step, we have shown that under the conditions specified in theorem, we have $\mu_Y \sim \mu_{Y|W=w} \sim \mu_B$.

\noindent {\bf Step $\mathbf{2}$.} For any $n$ and $\gamma(\cdot), \phi(\cdot) \in C[0, T]$, we follow~\cite{LS} and define a truncated version of $g$ as follows:
\begin{equation*}
g_{(n)}(t, \gamma_0^t, \phi_0^t)=g(t, \gamma_0^t, \phi_0^t) \mathbf{1}_{\int_0^t g^2(s, \gamma_0^t, \phi_0^s) ds < n}.
\end{equation*}
Now, define a truncated version of $Y$ as follows:
\begin{equation*}
Y_{(n)}(t)=\rho \int_0^t g_{(n)}(s, W_0^s, Y_0^s) ds +B(t), \quad t \in [0, T],
\end{equation*}
which, as elaborated on Page $265$ in~\cite{LS}, can be rewritten as
\begin{equation*}
Y_{(n)}(t)=\rho \int_0^t g_{(n)}(s, W_0^s, Y_{(n), 0}^s) ds +B(t), \quad t \in [0, T].
\end{equation*}
It is well known~\cite{du70, ka71} that
\begin{equation*}
I(W_{0}^T; Y_{(n), 0}^T)=\frac{\rho^2}{2} \int_0^T \E[g^2_{(n)}(s, W_0^s, Y_{(n), 0}^s)]-\E[\E^2[g_{(n)}(s, W_0^s, Y_{(n), 0}^s)|Y_{(n), 0}^s]] ds,
\end{equation*}
and
\begin{equation*}
I(W_{0}^T; Y_{0}^T) =\frac{\rho^2}{2} \int_0^T \E[g^2(s, W_0^s, Y_0^s)]-\E[\E^2[g(s, W_0^s, Y_0^s)|Y_{0}^s]] ds.
\end{equation*}
Moreover, it follows from Theorem~\ref{continuous-time-main-result-1} (here, note that extra yet minor care has to be taken since $g_{(n)}(s, W_0^s, Y_{(n), 0}^s)$ is only a piecewise differentiable function in $\rho$; cf. Condition (c)) that
\begin{align}
\nonumber \frac{d}{d\rho} I(W_{0}^T; Y_{(n), 0}^T) & =\rho \int_0^T \E[g^2_{(n)}(s, W_0^s, Y_{(n), 0}^s)]-\E[\E^2[g_{(n)}(s, W_0^s, Y_{(n), 0}^s)|Y_{(n), 0}^T]] ds \\
& \hspace{-5cm} +\rho^2 \int_0^T \E\left[g_{(n)}(s, W_0^s, Y_{(n), 0}^s) \E_W\left[\left.\frac{d}{d\rho} g_{(n)}(s, W_0^s, Y_{(n), 0}^s)\right|Y_{(n),0}^T\right] -\E[g_{(n)}(s, W_0^s, Y_{(n), 0}^s)|Y_{(n), 0}^T] \frac{d}{d\rho} g_{(n)}(s, W_0^s, Y_{(n), 0}^s)\right]ds \notag\\
& \hspace{-4cm}+\rho^2\int_0^T \E\left[g_{(n)}(s, W_0^s, Y_{(n),0}^s) \frac{d}{d\rho} g_{(n)}(s, W_0^s, Y_{(n),0}^s) \right]ds-\rho \E\left[\left.\E_W\left[\int_0^T \frac{d}{d\rho} g_{(n)}(s, W_0^s, Y_{(n),0}^s) dY_{(n)}(s)\right|Y_{(n), 0}^T \right]\right]. \label{truncated-derivative}
\end{align}

\noindent {\bf Step $\mathbf{3}$.} In this step, we will prove that
\begin{align}  \label{f-prime-convergent}
\nonumber \lim_{n \to \infty} \frac{d}{d\rho} I(W_{0}^T; Y_{(n), 0}^T) & = \rho \int_0^T \E[g^2(s, W_0^s, Y_0^s)]-\E[\E^2[g(s, W_0^s, Y_0^s)|Y_{0}^T]] ds\\
& \hspace{-4cm} +\rho^2 \int_0^T \E\left[g(s, W_0^s, Y_{0}^s) \E_W\left[\left.\frac{d}{d\rho} g(s, W_0^s, Y_{0}^s)\right|Y_{0}^T\right] -\E[g(s, W_0^s, Y_{0}^s)|Y_{0}^T] \frac{d}{d\rho} g(s, W_0^s, Y_{0}^s)\right]ds \notag\\
& \hspace{-4cm}+\rho^2\int_0^T \E\left[g(s, W_0^s, Y_{0}^s) \frac{d}{d\rho} g(s, W_0^s, Y_{0}^s) \right]ds-\rho \E\left[\left.\E_W\left[\int_0^T \frac{d}{d\rho} g(s, W_0^s, Y_{0}^s) dY(s)\right|Y_{0}^T \right]\right].
\end{align}

{\bf Step $\mathbf{3.1}$.} In this step, we observe that, with Condition (d), an application of the dominated convergence theorem will yield
\begin{equation*}
\lim_{n \to \infty} \int_0^T \E[g^2_{(n)}(s, W_0^s, Y_{(n), 0}^s)] ds = \int_0^T \E[g^2(s, W_0^s, Y_{0}^s)] ds.
\end{equation*}

{\bf Step $\mathbf{3.2}$.} In this step, we will prove that
\begin{equation}  \label{five-steps-2}
\lim_{n \to \infty} \int_0^T \E[\E^2[g_{(n)}(s, W_0^s, Y_{(n), 0}^s)|Y_{(n), 0}^T]] ds = \int_0^T \E[\E^2[g(s, W_0^s, Y_{0}^s)|Y_{0}^T]] ds.
\end{equation}
First of all, we note that
\begin{align*}
\E[\E^2[g_{(n)}(s, W_0^s, Y_{(n), 0}^s)|Y_{(n), 0}^T]]
&=\E\left[\left(\int g_{(n)}(s, w_0^s, Y_{(n), 0}^s) \mu_{W|Y_{(n)}}(dw|Y_{(n),0}^T)\right)^2\right]\\
&\hspace{-1cm} =\E\left[\left(\int g_{(n)}(s, w_0^s, Y_{(n), 0}^s) \frac{d\mu_{Y_{(n)}|W}}{d\mu_B}(Y_{(n),0}^T|w_0^T) \mu_W(dw) / \frac{d\mu_{Y_{(n)}}}{d\mu_B}(Y_{(n), 0}^T) \right)^2 \right]\\
&\hspace{-1cm} =\E \left[ \left(\int g_{(n)}(s, w_0^s, B_0^T) \frac{d\mu_{Y_{(n)}|W}}{d\mu_B}(B_0^T|w_0^T) \mu_W(dw) \right)^2 \times \left( \frac{d\mu_{Y_{(n)}}}{d\mu_B}(B_0^T) \right)^{-1} \right].
\end{align*}
We now proceed with the following steps:

{\bf Step $\mathbf{3.2.1}$.} In this step, we prove that in probability
\begin{equation*}
\frac{d\mu_{Y_{(n)}}}{d\mu_B}(B_0^T) \to \frac{d\mu_{Y}}{d\mu_B}(B_0^T).
\end{equation*}
First of all,
\begin{equation*}
\frac{d\mu_{Y_{(n)}}}{d\mu_B}(B_0^T) = \int \exp\left(\rho \int_0^T g_{(n)}(s, w_0^s, B_0^s) dB(s)-\frac{\rho^2}{2}\int_0^t g^2_{(n)}(s, w_0^s, B_0^s)ds\right) \mu_W(dw).
\end{equation*}
With Condition (d), we apply the It\^{o} isometry~\cite{ok95} to deduce that
\begin{equation*}
\E\left[\left(\int_0^T (g_{(n)}(s, w_0^s, B_0^s)-g(s, w_0^s, B_0^s)) dB(s)\right)^2\right]=\E\left[\int_0^T (g_{(n)}(s, w_0^s, B_0^s)-g(s, w_0^s, B_0^s))^2 ds\right] \to 0,
\end{equation*}
which further implies that
\begin{equation*}
\exp\left(\rho \int_0^T g_{(n)}(s, w_0^s, B_0^s) dB(s)-\frac{\rho^2}{2}\int_0^t g^2_{(n)}(s, w_0^s, B_0^s) ds\right)
\end{equation*}
converges to
\begin{equation*}
\exp\left(\rho \int_0^T g(s, w_0^s, B_0^s) dB(s)-\frac{\rho^2}{2}\int_0^t g^2(s, w_0^s, B_0^s) ds\right)
\end{equation*}
in probability. And moreover, it can be easily checked that
\begin{equation*}
\E\left[\int \exp\left(\rho \int_0^T g_{(n)}(s, w_0^s, B_0^s) dB(s)-\frac{\rho^2}{2}\int_0^t g^2_{(n)}(s, w_0^s, B_0^s) ds\right) \mu_W(dw)\right] = \E\left[\frac{d\mu_{Y_{(n)}}}{d\mu_B}(B_0^T)\right]=1
\end{equation*}
and
\begin{equation*}
\E\left[\int \exp\left(\rho \int_0^T g(s, w_0^s, B_0^s) dB(s)-\frac{\rho^2}{2}\int_0^t g^2(s, w_0^s, B_0^s) ds\right) \mu_W(dw)\right] = \E\left[\frac{d\mu_{Y}}{d\mu_B}(B_0^T)\right] =1.
\end{equation*}
It then follows from Theorem $5.5.2$ of~\cite{Durrett} that
\begin{equation*}
\lim_{n \to \infty} \E\left[\int \left|\left(\exp\left(\rho \int_0^T g_{(n)}(s, w_0^s, B_0^s) dB(s)-\frac{\rho^2}{2}\int_0^t g^2_{(n)}(s, w_0^s, B_0^s) ds\right) \right. \right. \right.
\end{equation*}
\begin{equation*}
\left. \left. \left. -\exp\left(\rho \int_0^T g(s, w_0^s, B_0^s) dB(s)-\frac{\rho^2}{2}\int_0^t g^2(s, w_0^s, B_0^s) ds\right) \right) \right| \mu_W(dw)\right]= 0,
\end{equation*}
which further implies that
\begin{equation*}
\int \exp\left(\rho \int_0^T g_{(n)}(s, w_0^s, B_0^s) dB_s-\frac{\rho^2}{2}\int_0^t g^2_{(n)}(s, w_0^s, B_0^s) ds\right) \mu_W(dw)
\end{equation*}
converges to
\begin{equation*}
\int \exp\left(\rho \int_0^T g(s, w_0^s, B_0^s) dB_s-\frac{\rho^2}{2}\int_0^t g^2(s, w_0^s, B_0^s) ds\right) \mu_W(dw)
\end{equation*}
in probability.

{\bf \phantom{1234} Step $\mathbf{3.2.2}$.} In this step, we will prove that in probability
\begin{equation*}
\int g_{(n)}(s, w_0^s, B_0^s) \frac{d\mu_{Y_{(n)}|W}}{d\mu_B}(B_0^T|w_0^T) \mu_W(dw) \to \int g(s, w_0^s, B_0^s) \frac{d\mu_{Y|W}}{d\mu_B}(B_0^T|w_0^T) \mu_W(dw).
\end{equation*}
First of all, it is easy to check that in probability
\begin{equation*}
g_{(n)}(s, w_0^s, B_0^s) \frac{d\mu_{Y_{(n)}|W}}{d\mu_B}(B_0^T|w_0^T) \to g(s, w_0^s, B_0^s) \frac{d\mu_{Y|W}}{d\mu_B}(B_0^T|w_0^T).
\end{equation*}
And moreover, we have
\begin{equation*}
\E\left[\int |g_{(n)}(s, w_0^s, B_0^s)| \frac{d\mu_{Y_{(n)}|W}}{d\mu_B}(B_0^T|w_0^T) \mu_W(dw)\right]=\E[|g_{(n)}(s, W(s), Y_{(n),0}^s)|]
\end{equation*}
converges to
\begin{equation*}
\E[|g(s, W_0^s, Y_{0}^s)|]=\E\left[\int |g(s, w_0^s, B_0^s)| \frac{d\mu_{Y|W}}{d\mu_B}(B_0^T|w_0^T) \mu_W(dw)\right].
\end{equation*}
So, similarly as in Step $3.1.1$, we deduce that
\begin{equation*}
\int g_{(n)}(s, w_0^s, B_0^s) \frac{d\mu_{Y_{(n)}|W}}{d\mu_B}(B_0^T|w) \mu_W(dw)
\to \int g(s, w_0^s, B_0^s) \frac{d\mu_{Y|W}}{d\mu_B}(B_0^T|w) \mu_W(dw).
\end{equation*}
in probability.

{\bf Step $\mathbf{3.2.3}$}. Note that Steps $3.2.1$ and $3.2.2$ collectively yield that
\begin{equation*}
 \left(\int g_{(n)}(s, w_0^s, B_0^T) \frac{d\mu_{Y_{(n)}|W}}{d\mu_B}(B_0^T|w_0^T) \mu_W(dw) \right)^2 \times \left( \frac{d\mu_{Y_{(n)}}}{d\mu_B}(B_0^T) \right)^{-1}
\end{equation*}
converges to
\begin{equation*}
\left(\int g(s, w_0^s, B_0^T) \frac{d\mu_{Y|W}}{d\mu_B}(B_0^T|w_0^T) \mu_W(dw) \right)^2 \times \left( \frac{d\mu_{Y}}{d\mu_B}(B_0^T) \right)^{-1}
\end{equation*}
in probability. Now, applying Jensen's inequality, we have
\begin{align*}
&\hspace{-1cm} \left(\int g_{(n)}(s, w_0^s, B_0^T) \frac{d\mu_{Y_{(n)}|W}}{d\mu_B}(B_0^T|w_0^T) \mu_W(dw) \right)^2 \times \left( \frac{d\mu_{Y_{(n)}}}{d\mu_B}(B_0^T) \right)^{-1}\\
&=\left(\int g_{(n)}(s, w_0^s, B_0^T) \frac{d\mu_{Y_{(n)}|W}}{d\mu_B}(B_0^T|w_0^T) \mu_W(dw)/\frac{d\mu_{Y_{(n)}}}{d\mu_B}(B_0^T)\right)^2 \times \frac{d\mu_{Y_{(n)}}}{d\mu_B}(B_0^T)\\
& \leq \left(\int g_{(n)}^2(s, w_0^s, B_0^T) \frac{d\mu_{Y_{(n)}|W}}{d\mu_B}(B_0^T|w_0^T) \mu_W(dw)/\frac{d\mu_{Y_{(n)}}}{d\mu_B}(B_0^T)\right) \times \frac{d\mu_{Y_{(n)}}}{d\mu_B}(B_0^T)\\
& = \int g_{(n)}^2(s, w_0^s, B_0^T) \frac{d\mu_{Y_{(n)}|W}}{d\mu_B}(B_0^T|w_0^T) \mu_W(dw).
\end{align*}
Note that
\begin{equation*}
\E\left[\int g_{(n)}^2(s, w_0^s, B_0^T) \frac{d\mu_{Y_{(n)}|W}}{d\mu_B}(B_0^T|w_0^T) \mu_W(dw)\right]=\E[g_{(n)}^2(s, W_0^s, Y_{(n), 0}^s)] \to \E[g^2(s, W_0^s, Y_0^s)] < \infty,
\end{equation*}
where the finiteness is due to Condition (d). Finally, the desired (\ref{five-steps-2}) follows from the generalized dominated convergence theorem (see, e.g., Theorem $19$ on Page $89$ of~\cite{ro10}).

{\bf Step $\mathbf{3.3}$.} In this step, we establish the following two convergences:
\begin{equation}  \label{extra-term-1}
\hspace{-2.3cm} \lim_{n \to \infty} \int_0^T \E\left[g_{(n)}(s, W_0^s, Y_{(n), 0}^s)\E_W\left[\left.\frac{d}{d\rho} g_{(n)}(s, W_0^s, Y_{(n), 0}^s)\right|Y_{(n),0}^T\right]\right]ds= \int_0^T \E\left[g(s, W_0^s, Y_{0}^s) \E_W\left[\left.\frac{d}{d\rho} g(s, W_0^s, Y_{0}^s)\right|Y_0^T\right]\right]ds
\end{equation}
and
\begin{equation}  \label{extra-term-2}
\hspace{-1cm} \lim_{n \to \infty} \int_0^T \E\left[\E[g_{(n)}(s, W_0^s, Y_{(n), 0}^s)|Y_{(n), 0}^T] \frac{d}{d\rho} g_{(n)}(s, W_0^s, Y_{(n), 0}^s)\right]ds = \int_0^T \E\left[\E[g(s, W_0^s, Y_{0}^s)|Y_0^T] \frac{d}{d\rho} g(s, W_0^s, Y_{0}^s)\right]ds,
\end{equation}
and
\begin{equation} \label{extra-term-3}
\lim_{n \to \infty} \int_0^T \E\left[g_{(n)}(s, W_0^s, Y_{(n),0}^s) \frac{d}{d\rho} g_{(n)}(s, W_0^s, Y_{(n),0}^s) \right]ds = \int_0^T \E\left[g(s, W_0^s, Y_{0}^s) \frac{d}{d\rho} g(s, W_0^s, Y_{0}^s) \right]ds,
\end{equation}
and
\begin{equation} \label{extra-term-4}
\lim_{n \to \infty} \E\left[\left.\E_W\left[\int_0^T \frac{d}{d\rho} g_{(n)}(s, W_0^s, Y_{(n),0}^s) dY_{(n)}(s)\right|Y_{(n), 0}^T \right]\right] = \E\left[\left.\E_W\left[\int_0^T \frac{d}{d\rho} g(s, W_0^s, Y_{0}^s) dY(s)\right|Y_{0}^T \right]\right].
\end{equation}

{\bf Step $\mathbf{3.3.1}$.} In this step, we will prove (\ref{extra-term-3}). Writing $g_{(n)}(s, W_0^s, Y_{(n), 0}^s), g(s, W_0^s, Y_0^s)$ as $g_{(n)}(s), g(s)$ for notational simplicity, we have
\begin{equation*}
\int_0^T \E\left[g_{(n)}(s)\frac{d}{d\rho} g_{(n)}(s)\right]ds - \int_0^T \E\left[g(s) \frac{d}{d\rho} g(s)\right]ds
\end{equation*}
\begin{equation*}
\hspace{-1cm} =\int_0^T \E\left[g_{(n)}(s)\frac{d}{d\rho} g_{(n)}(s)\right]ds - \int_0^T \E\left[g(s) \frac{d}{d\rho} g_{(n)}(s)\right]ds+ \int_0^T \E\left[g(s) \frac{d}{d\rho} g_{(n)}(s)\right]ds- \int_0^T \E\left[g(s) \frac{d}{d\rho} g(s)\right]ds
\end{equation*}
\begin{equation*}
=\int_0^T \E\left[(g_{(n)}(s)-g(s))\frac{d}{d\rho} g_{(n)}(s)\right]ds - \int_0^T \E\left[g(s) \left(\frac{d}{d\rho} g_{(n)}(s)-\frac{d}{d\rho} g(s)\right)\right]ds.
\end{equation*}
The desired convergences then follow from the fact that as $n$ tends to infinity,
\begin{equation*}
\left(\int_0^T \E \left[(g_{(n)}(s)-g(s))\frac{d}{d\rho} g_{(n)}(s)\right] ds \right)^2 \leq \int_0^T \E[(g_{(n)}(s)-g(s))^2] ds \int_0^T \E\left[\left(\frac{d}{d\rho} g_{(n)}(s)\right)^2\right] ds \to 0
\end{equation*}
and
\begin{equation*}
\left(\int_0^T \E\left[g(s) \left(\frac{d}{d\rho} g_{(n)}(s)-\frac{d}{d\rho} g(s)\right)\right] ds \right)^2 \leq \int_0^T \E[g^2(s)] ds \int_0^T \E\left[\left(\frac{d}{d\rho} g_{(n)}(s)-\frac{d}{d\rho} g(s)\right)^2 \right] ds \to 0,
\end{equation*}
where we have used the fact that for any $n$,
\begin{equation} \label{no-point-mass}
\frac{d}{d\rho} g_{(n)}(s, W_0^s, Y_{(n), 0}^s)=\left(\frac{d}{d\rho} g(s, W_0^s, Y_{0}^s)\right) \mathbf{1}_{\int_0^s g^2(t, W_0^t, Y_{0}^t) dt < n} \quad \mbox{a.s.},
\end{equation}
which is implied by Condition (f).

{\bf Step $\mathbf{3.3.2}$.} In this step, we will prove (\ref{extra-term-1}), (\ref{extra-term-2}) and (\ref{extra-term-4}), which all follow from similar arguments as in Step $3.2$.

{\bf Step $\mathbf{3.4}$.} Note that Steps $3.1$, $3.2$ and $3.3$ collectively yield (\ref{f-prime-convergent}).

\noindent {\bf Step $\mathbf{4}$.} In this step, we will prove
\begin{equation}  \label{f-convergent}
\lim_{n \to \infty} I(W_{0}^T; Y_{(n), 0}^T) = I(W_{0}^T; Y_{0}^T).
\end{equation}
Obviously, it suffices to prove that
\begin{equation}  \label{f1-convergent}
\lim_{n \to \infty} \int_0^T \E[g^2_{(n)}(s, W_0^s, Y_{(n), 0}^s)] ds = \int_0^T \E[g^2(s, W_0^s, Y_0^s)] ds,
\end{equation}
and
\begin{equation}  \label{f2-convergent}
\lim_{n \to \infty} \int_0^T \E[\E^2[g_{(n)}(s, W_0^s, Y_{(n),0}^s)|Y_{(n), 0}^s]] ds = \int_0^T \E[\E^2[g(s, W_0^s, Y_0^s)|Y_{0}^s]] ds.
\end{equation}
Note that (\ref{f1-convergent}) has been established in Step $3.1$, and the proof of (\ref{f2-convergent}) can be established using a parallel argument as in Step $3.2$.

\noindent {\bf Step $\mathbf{5}$.} In this step, we will establish the continuity of the following terms with respect to $\rho$:
$$
\int_0^T \E[g^2(s, W_0^s, Y_0^s)] ds, \quad \int_0^T \E[\E^2[g(s, W_0^s, Y_0^s)|Y_{0}^T]] ds,
$$
$$
\int_0^T \E\left[g(s, W_0^s, Y_{0}^s) \frac{d}{d\rho} g(s, W_0^s, Y_{0}^s) \right]ds, \quad \int_0^T \E\left[\E[g(s, W_0^s, Y_0^s)|Y_0^T] \frac{d}{d\rho} g(s, W_0^s, Y_0^s)\right]ds,
$$
\begin{equation}  \label{another-six-quantities}
\int_0^T \E\left[g(s, W_0^s, Y_0^s) \E_W\left[\left.\frac{d}{d\rho} g(s, W_0^s, Y_0^s)\right|Y_0^T\right]\right]ds, \quad \E\left[\left.\E_W\left[\int_0^T \frac{d}{d\rho} g(s, W_0^s, Y_{0}^s) dY(s)\right|Y_{0}^T \right]\right].
\end{equation}

Note that the continuity of $\int_0^T \E[g^2(s, W_0^s, Y_0^s)] ds$ immediately follows from the dominated convergence theorem together with Condition (d) and the fact that $g(s, W_0^s, Y_0^s)$ is continuous in $\rho$. And moreover, a parallel argument can be used to establish the continuity of
$$
\int_0^T \E\left[g(s, W_0^s, Y_0^s) \frac{d}{d\rho} g(s, W_0^s, Y_0^s)\right]ds.
$$

To establish the continuity of $\int_0^T \E[\E^2[g(s, W_0^s, Y_0^s)|Y_{0}^T]] ds$, it suffices to prove that for any sequence $\{\rho_n\}$ convergent to $\rho$,
\begin{equation*}
\lim_{n \to \infty} \E[\E^2[g(s, W_0^s, Y_0^{(\rho_n), s})|Y_{0}^{(\rho_n), T}]] =
\int_0^T \E[\E^2[g(s, W_0^s, Y_0^{(\rho), s})|Y_{0}^{(\rho), T}]] ds,
\end{equation*}
which can be shown in a parallel argument as in Step $3.2$, where the following similar convergence is proven:
$$
\lim_{n \to \infty} \int_0^T \E[\E^2[g_{(n)}(s, W_0^s, Y_{(n), 0}^s)|Y_{(n), 0}^T]] ds=\int_0^T \E[\E^2[g(s, W_0^s, Y_0^s)|Y_{0}^T]] ds.
$$
Furthermore, similarly as in Step $3.3.2$, the continuity of other quantities in (\ref{another-six-quantities})
can be established as well.

\noindent {\bf Step $\mathbf{6}$.} It then follows from (\ref{truncated-derivative}) that, for any $\tau > 0$,
{\small \begin{align}
I(W_{0}^T; Y_{(n), 0}^{(\tau), T}) &= \int_0^{\tau} \frac{d}{d\rho} I(W_{0}^T; Y_{(n), 0}^{(\rho), T}) d\rho \notag\\
& = \rho \int_0^{\tau}\int_0^T \E[g_{(n)}^2(s, W_0^s, Y_{(n), 0}^s)]-\E[\E^2[g_{(n)}(s, W_0^s, Y_{(n), 0}^s)|Y_{0}^T]] ds d\rho \notag\\
& \hspace{-4cm} +\rho^2 \int_0^{\tau} \int_0^T \E\left[g_{(n)}(s, W_0^s, Y_{(n), 0}^s) \E_W\left[\left.\frac{d}{d\rho} g_{(n)}(s, W_0^s, Y_{(n), 0}^s)\right|Y_{(n), 0}^T\right] -\E[g_{(n)}(s, W_0^s, Y_{(n), 0}^s)|Y_{0}^T] \frac{d}{d\rho} g_{(n)}(s, W_0^s, Y_{(n), 0}^s)\right]ds d\rho \notag\\
& \hspace{-4cm}+\rho^2 \int_0^{\tau} \int_0^T \E\left[g_{(n)}(s, W_0^s, Y_{(n), 0}^s) \frac{d}{d\rho} g_{(n)}(s, W_0^s, Y_{(n), 0}^s) \right]ds d\rho -\rho \int_0^{\tau} \E\left[\left.\E_W\left[\int_0^T \frac{d}{d\rho} g_{(n)}(s, W_0^s, Y_{(n), 0}^s) dY_{(n)}(s)\right|Y_{(n), 0}^T \right]\right] d\rho, \nonumber
\end{align}}
where we have used the superscripts $(\rho)$ and $(\tau)$ to specify the underlying parameters. Applying the dominated convergence theorem, we have
{\small \begin{align} \label{integral-version}
I(W_{0}^T; Y_{0}^{(\tau), T}) & = \rho \int_0^{\tau}\int_0^T \E[g^2(s, W_0^s, Y_{0}^s)]-\E[\E^2[g(s, W_0^s, Y_{0}^s)|Y_{0}^T]] ds d\rho \notag\\
& \hspace{-4cm} +\rho^2 \int_0^{\tau} \int_0^T \E\left[g(s, W_0^s, Y_{0}^s) \E_W\left[\left.\frac{d}{d\rho} g(s, W_0^s, Y_{0}^s)\right|Y_{0}^T\right] -\E[g(s, W_0^s, Y_{0}^s)|Y_{0}^T] \frac{d}{d\rho} g(s, W_0^s, Y_{0}^s)\right]ds d\rho \notag\\
& \hspace{-4cm}+\rho^2 \int_0^{\tau} \int_0^T \E\left[g(s, W_0^s, Y_{0}^s) \frac{d}{d\rho} g(s, W_0^s, Y_{0}^s) \right]ds d\rho -\rho \int_0^{\tau} \E\left[\left.\E_W\left[\int_0^T \frac{d}{d\rho} g(s, W_0^s, Y_{0}^s) dY(s)\right|Y_{0}^T \right]\right] d\rho.
\end{align}}
Note that Step $5$ has established the continuity of the integrand (with respect to $d \rho$) at the RHS of (\ref{integral-version}). So, the desired formula (\ref{continuous-time-feedback-formula-2}) then follows from taking the derivative of (\ref{integral-version}) with respect to $\tau$, and the proof of the theorem is then complete.

\begin{rem}  \label{rigorously-recover}
Theorem~\ref{continuous-time-I-MMSE} is ``essentially'' included by Theorem~\ref{continuous-time-main-result-2} as a special case. More precisely, under the assumptions of Theorem~\ref{continuous-time-I-MMSE}, the average power constraint (\ref{square-integrable}) trivially implies Conditions (b), (c) and (d). Note that in the proof of Theorem~\ref{continuous-time-main-result-2}, the sole use of Condition (f) is deducing (\ref{no-point-mass}), a weaker yet somewhat cumbersome condition, which is also implied by (\ref{square-integrable}). So, with Condition (f) replaced by (\ref{no-point-mass}), Theorem~\ref{continuous-time-main-result-2} recovers Theorem~\ref{continuous-time-I-MMSE} with a direct and rigorous proof~\footnote{For sticklers demanding mathematical rigor and perfection: It is known that there are multiple ``missing steps'' in the proof of Theorem~\ref{continuous-time-I-MMSE} in~\cite{gu05}. For instance, the differentiability of $I(X_0^T; Y_0^T)$ with respect to $snr$ does not seem to be trivial and thereby demands careful justifications, which are however absent in~\cite{gu05}; also, from (259) to (270) in the proof of Lemma $5$ (a key lemma for the proof of Theorem~\ref{continuous-time-I-MMSE}), the authors assumed that for a sequence of random variable $X_n$ convergent to $0$ almost surely, $\lim_{n \to \infty} \E[X_n]=0$, which is not true in general.}.
\end{rem}

\begin{rem}
To show (\ref{five-steps-2}), as opposed to our approach in Step $3.2$, a possible and seemingly more natural first step is to establish the convergence of $\E^2[g_{(n)}(s, W_0^s, Y_{(n), 0}^s)|Y_{(n), 0}^T]$ (either in probability or distribution) as $n$ tends to infinity, which, however, has eluded our multiple attempts. Note that for the above-mentioned convergence, the martingale convergence theorem may not be applied, since it is not clear if the $\sigma$-algebra generated by $Y_{(n), 0}^T$ gets larger at $n$ increases. Similar hurdles were encountered in our attempts to prove (\ref{extra-term-2}) and (\ref{f2-convergent}), and parallel arguments as in Step $3.2$ have to be used instead. Here, we remark that, in general, the problem of establishing the convergence of a sequence of conditional expectations can be rather subtle and challenging; see some positive results in~\cite{go94} and~\cite{cl05} where some fairly strong assumptions are imposed.
\end{rem}

\section{Possible Future Directions}  \label{outlook}

The significant impact of the original I-MMSE relationship (\ref{I-MMSE}) on non-feedback/memoryless channels presages many possible applications of the extended I-MMSE relationships (\ref{extended-formula-feedback}), (\ref{extended-formula-memory}), (\ref{continuous-extended-formula-feedback}), (\ref{continuous-extended-formula-memory}) to situations where the feedback/memory are present; moreover, we envision that our new approach can provide new perspectives to examine a number of aspects in information theory. In this section, we will discuss some promising future directions one can further pursue based on this work. In a nutshell, the possible further directions can be summarized as follows:
\begin{enumerate}
\item further extend the I-MMSE relationship to colored Gaussian feedback channels, general feedback channels, its limiting version in terms of mutual information rate, extensions with relaxed assumptions;
\item explore the properties of the extended MMSE;
\item explore the applications of the extended I-MMSE relationship to Gaussian feedback channels, multi-user Gaussian channels, Gaussian channels with input/output memory;
\item explore the applications of our new approach to other information-theoretic quantities, higher order derivatives, entropy power inequalities, sampling theorems, and so on.
\end{enumerate}

\subsection{Further Extensions of the I-MMSE relationship} \label{further-extensions}

\indent {\bf Colored Gaussian feedback channels.} The discrete-time I-MMSE relationship (\ref{I-MMSE}) carries over verbatim to linear vector Gaussian channels~\cite{gu05}, and its extensions to more general settings include derivatives with respect to arbitrary parameterizations~\cite{pa06}, higher order derivatives~\cite{pa09}, and so on. Extensions of the continuous-time I-MMSE relationship (\ref{continuous-I-MMSE}) have been studied as well; representative work include fractional Brownian motion noise~\cite{du08} and an abstract Wiener space~\cite{za05, wo07}. On the other hand, all the above-mentioned extensions have been confined to the scenarios where the feedback are absent.

In view of our results on extensions of the I-MMSE relationship, one of the possible future directions is to further extend the I-MMSE relationship to colored Gaussian feedback channels in both discrete time and continuous time.

While the proposed direction is well within reach in discrete time, the same problem appears to be far more challenging in continuous time due to the inherent intractability of continuous-time Gaussian processes. A natural goal in this direction is to find the broadest class of continuous-time Gaussian processes for which the extended I-MMSE relationship holds. One special class of Gaussian processes that appear to be tractable are those featuring canonical representations~\cite{hi93} (in terms of the standard Brownian motions) without discrete spectrum terms (see (6.8.2) of~\cite{ih93}), and thereby Girsanov's theorem~\cite{LS}, a key technical ingredient used in our proofs of Theorems~\ref{continuous-time-main-result-1} and~\ref{continuous-time-main-result-2}, can be carried over to such processes. Since fractional Brownian motions are a special class of such Gaussian processes, one would arrive at results which include the ones in~\cite{du08} as special cases.

\indent {\bf General feedback channels.} The exploration of fundamental relationships between information and estimation measures has not been confined to Gaussian channels only. As a matter of fact, a considerable amount of work, largely inspired by the I-MMSE relationship for Gaussian channels, have been devoted to investigating non-Gaussian channels for parallel relationships. In this direction, representative work include additive channels~\cite{gu05a}, arbitrary channels~\cite{pa07}, Poisson channels~\cite{gu08, at12, ta14a}, binomial and negative binomial channels~\cite{ta14, ta14a}. This thread of efforts have culminated in a recent paper~\cite{ji14}, where a unified general formula relating information and estimation measures was derived for L\'evy channels, which encompass Gaussian channels and a number of other non-Gaussian channels as special cases.

One of the possible directions is to further generalize the result in~\cite{ji14} to Levy channels with feedback/memory, in either discrete or continuous time. Alternatively, one can also consider deriving the extended I-MMSE relationship for channel featuring noise with jumps (obviously, noise of this type naturally exists in a variety of real-life situations). For this direction, it might be wiser to first consider additive Levy processes (which are different from Levy channels in~\cite{ji14} in spite of the same name), which have been extensively studied in mathematical theory and practical applications. Note that such extension, if successful, would generalize the one in~\cite{du10}, which only deals with pure jump processes. A key ingredient for success would be an ``explicit'' Girsanov-type theorem for Levy processes.

{\bf Limiting version.} For most non-degenerate channels with feedback/memory, the capacity is computed via maximizing the (directed) mutual information rate, rather than the mutual information. This fact necessitates the consideration of the limiting version of the extended I-MMSE relationship in discrete time as $n$ tends to infinity. The power of such a ``limiting approach'' has been showcased in Kim's variational formulation~\cite{ki10} of Gaussian feedback capacity as a limiting version of finite block capacity of Cover and Pombra~\cite{co1989}, which has been used to derived/expressed classes of Gaussian feedback channels. It is certainly worthwhile to explore whether the extended MMSE also feature a limiting version.

There are hurdles for the journey along this direction: First of all, not all input processes will guarantee the limit of the mutual information rate is well-defined. Another issue is the differentiability/smoothness/analyticity of the mutual information rate, which may fail for certain channels~\cite{gm05, hm10}. So, it makes senses to focus one's attention on identifying channels with explicit and reasonable assumptions on the input process for the existence of the mutual information rate and its derivative.

Probably a feasible first step is to examine Gaussian channels with Gaussian inputs and linear feedback. Such a coding scheme proves to be capacity achieving~\cite{co1989} and has been instrumental in Kim's variational formulation~\cite{ki10} of Gaussian feedback capacity. Alternatively, one can also consider Gaussian channels with inputs that feature certain Markovian strcture: at least for discrete-time Gaussian channels with ARMA noise, the capacity will be achieved by feedback-dependent Markovian input processes~\cite{ya07}, which makes it possible to apply Tatikonda's feedback capacity formulation and the corresponding dynamic programming approach~\cite{ta09}. Moreover, for certain Gaussian channels with certain Markovian inputs, the analyticity/smoothness/asymptotics of the mutual information rate has been established~\cite{hm10}.

{\bf Extensions with relaxed assumptions.} The I-MMSE relationship (\ref{I-MMSE}) was originally proven under the condition that the input has finite power, i.e., $\E[X^2] < \infty$, which Wu and Verdu~\cite{wu12} have recently weakened to the existence of the mutual information. Naturally, for the proposed feedback/memory extensions in this paper, one may explore whether/to what extent the finite power constraint can be relaxed. Such a task seems to be technically non-trivial, at least for the continuous-time setting. We expect that conditions that lead to some kind of weak continuity of the extended MMSE should be established for such a relaxation, as done in non-feedback case~\cite{wu12} in discrete time.

\subsection{Properties of the Extended MMSE}

Properties of the discrete-time MMSE associated with Gaussian non-feedback/memoryless channels, such as monotonicity, continuity, smoothness, analyticity, concavity and asymptotics, have been extensively studied~\cite{gu11, wu12}. These properties have been utilized in a wide range of applications; in particular, the following two properties~\cite{gu11} of the MMSE (and their extensions) are of great interest and of direct use in deriving the capacity regions of some multi-user Gaussian channels, such as Gaussian wiretap channels~\cite{bu09} and Gaussian broadcast channels~\cite{bu10, bu13}:
\begin{itemize}
\item Gaussian inputs are the hardest to estimate, which means that any non-Gaussian input yields strictly smaller MMSE than a Gaussian input of the same variance;
\item The single-crossing property, which, roughly speaking, says that a Gaussian MMSE curve (with respect to the $snr$) intersects with a non-Gaussian MMSE curve at most once.
\end{itemize}
Naturally one may consider exploring whether or to what extent these properties hold for the extended MMSE in both discrete and continuous time. It is clear that for the extended MMSE, whether these two properties will hold depends on the adopted encoding schemes (namely, $X$ in (\ref{extended-formula-feedback}) and (\ref{continuous-extended-formula-feedback})) or the types of Gaussian channels (equivalently, $g$ in (\ref{extended-formula-memory}) and (\ref{continuous-extended-formula-memory})), which points out a natural future direction: to explore in what scenarios these two properties hold for the extended MMSE. In this direction, some reasonable candidates include Gaussian channels with linear feedback encoding schemes (see, e.g.,~\cite{sc66, ih93}) or Gaussian channels with linear inter-symbol interference (see, e.g.,~\cite{hi88}).

\subsection{Applications to Colored Gaussian Feedback Channels} \label{application-1}

Despite extensive efforts spent on colored Gaussian feedback channels, the capacity of such channels has largely remained unknown, except for some special cases~\cite{ki10}. The extended I-MMSE relationships may be helpful to deepen our understanding of colored Gaussian feedback channels: First, notice that an application of the Cauchy-Schwarz  inequality yields that the correctional terms of an extended MMSE can be upper bounded by the MMSE term, up to a multiplicative constant. Since the MMSE term ``corresponds'' to Gaussian channels without feedback, it is plausible to at least derive some bound~\cite{eb70} (which may depend on the signal-to-noise ratio) between the ratio of the feedback capacity and non-feedback capacity. Second, written as the sum of an MMSE term and correctional terms, an extended MMSE can be of great help, in both discrete and continuous time, to describe the asymptotical behavior~\cite{de89a} of the feedback capacity for the regime when $snr$ is small or large.

While deriving the capacity of a general colored Gaussian feedback channel seems to be far-fetched, one may consider making use of the extended MMSE relationships to derive the feedback capacity for some special colored Gaussian feedback channels. It is well known (see, e.g., Cover and Pombra~\cite{co1989} and Ihara~\cite{ih93}) that for colored Gaussian feedback channels with average power constraints, linear feedback schemes with Gaussian inputs are sufficient to achieve the capacity. This fact can be a major boost of the chance of deriving the exact capacity using the extended I-MMSE: under a linear feedback encoding scheme, the inputs and the outputs are de facto jointly Gaussian, which means both the MMSE and the correctional terms can be more explicitly computed.

Other than Gaussian feedback channels with the average power constraint, the extended I-MMSE relationship for colored Gaussian feedback channels can help us to understand Gaussian channels with other constraints, such as peak power constraints~\cite{sm71} and finite-type input constraints~\cite{ZehaviWolf88,lm95,mrs98,hm09}. Though the exact capacity of such channels are extremely challenging, some bounds or asymptotics of the capacity seems to be well within reach given a corresponding extended I-MMSE.

The above-mentioned initiatives can be parallelly taken for multiple-input-multiple-output (MIMO) Gaussian feedback channels possibly with fading, either in discrete or continuous time. Note that discrete-time MIMO Gaussian fading non-feedback channels have been extensively studied; see representative work~\cite{Telatar1999, tse2002} on the capacity of such channels. With regard to extending the I-MMSE relationship to such channels in continuous-time, we remark that higher dimensional Girsanov's theorem still holds true and it appears that the extended I-MMSE relationship is within reach. Needlessly to say, such an extended I-MMSE relationship can offer a new perspective to examine discrete-time MIMO Gaussian fading feedback channels and further study continuous-time Gaussian MIMO feedback channels.

\subsection{Applications to Multi-User Gaussian Channels}

{\bf Discrete-time.} The original I-MMSE relationship has been applied to discrete-time multi-user non-feedback Gaussian channels including Gaussian broadcast channels, wiretap channels and interference channels and so on. Naturally, one tempting direction is to explore the possible applications of the extended I-MMSE relationship to discrete-time multi-user Gaussian channels when the feedback is present. For this purpose, one of the imminent problems is to identify those multi-user Gaussian channels for which linear feedback coding schemes achieve the capacity regions. Alternatively, one can also look into whether a ``multi-user'' version of the extended I-MMSE relationship exists, which may involve conditional mutual information with multiple message sets. As might be expected, such a multi-user extended I-MMSE relationship can provide more insights between the interactions among the users.

{\bf Continuous-time.} Recently, the infinite bandwidth capacity regions of a continuous-time white Gaussian multiple access channel with/without feedback, a continuous-time white Gaussian interference channel without feedback and a continuous-time white Gaussian broadcast channel without feedback have been derived in~\cite{li14}. The continuous-time I-MMSE relationship has been applied to derive the capacity region of continuous-time white Gaussian broadcast channels. It is very natural to further extend the above-mentioned results and derive the capacity region for more general Gaussian multi-user channels with feedback, such formulas might be of great help for the derivation of the capacity region of continuous-time white Gaussian broadcast channels with feedback, or even more general continuous-time multi-user channels.

\subsection{Applications to Gaussian Memory Channels}

It is conceivable that the extended I-MMSE relationships (\ref{extended-formula-memory}) and (\ref{continuous-extended-formula-memory}) may be helpful for us to further understand Gaussian memory channels (see, e.g.,~\cite{sh91, me14}). To be more precise, we believe that such extended relationships will be helpful in terms of estimating/computing the capacity (region) of (multi-user) Gaussian channels with input/output memory.

\subsection{Applications of Our New Approach}

Other than the extended I-MMSE relationships, one may also consider whether/how the proposed new approach for deriving the extended I-MMSE relationship can be applied elsewhere. Below is a list of several scenarios where it can be instrumental.

{\bf Other information-theoretic quantities.} Other than recovering and extending the original I-MMSE relationship, the proposed approach in this paper may be further applied to study other information-theoretic quantities as well, which has been evidenced by the simple and direct proof (see Section~\ref{new-proof-2}) for the classical de Bruijn's identity~\cite{st59, co85}. It is our opinion that investigations on whether our approach can be applied elsewhere, particularly to the situations where the derivatives of certain information-theoretic quantity are needed, is highly likely to bear fruit. Here, we remark that the derivative of relative entropy has been examined for channels involving mismatched estimation without feedback/memory via different approaches from ours (see~\cite{ve09, we10}); it is possible that our approach in this work may provide an alternative proof or even extend the obtained results to more general channels with feedback and memory.

{\bf Higher order derivatives.} The second order derivative of the mutual information and entropy power function have been computed in~\cite{gu11, pa09}, which, among many other applications, have played a key role in understanding the concavity of the mutual information and deriving entropy power inequalities for Gaussian channels~\cite{co85, de89, pa09, pa11} and deriving the ``single crossing property'' of MMSE~\cite{gu11}; very recently, higher order derivatives of MMSE have been computed in a recursive form via a special Lie algebra structure~\cite{le15}. We expect that such results can be extended to Gaussian feedback channels. Rough computations suggest that the framework of our approach can also be applied to compute higher order derivatives explicitly. Other than understanding concavity, such explicit expressions can also help to characterize the asymptotic behavior of the mutual information and entropy power function associated with Gaussian feedback channels. In this direction, some Talyor-series-expansion-like formulae seem to be within reach, which, undoubtedly, will yield a finer characterization of the behavior of the mutual information and entropy power function of Gaussian feedback channels.

{\bf Entropy power inequalities.} The ideas and techniques in the proof of the original I-MMSE relationship has been used to give new and simpler proofs of a number of entropy power inequalities~\cite{ve06} associated with Gaussian non-feedback channels. It is certainly worthwhile to look into whether these inequalities can be extended to Gaussian feedback channels using our new approach. And, obviously, the same questions can be asked in the continuous-time setting, which, however, appears to be much more challenging.

{\bf Sampling theorems.} The general framework and technical tools employed in this work may further be used to explore the connections between discrete-time and continuous-time channels, or more precisely, whether/how one continuous-time channel can be approximated by its discrete-time versions obtained through sampling. In this direction, we propose the following conjecture.
\begin{conj} \label{discrete-to-continuous}
Consider the continuous-time channel (\ref{feedback-sampling-conjecture}). Let $0=t_1 < t_2 < \cdots < t_{m-1} < t_m=T$ and define $\triangle_m=\max \{t_i-t_{i-1}: i=2, 3,\ldots, m\}$. Assume the input $X(s, M, Y_0^s)$ is continuous in $s \in [0, T]$ and $\int_0^T \E[X^2(s, M, Y_0^s)] ds < \infty$. Then, we have
$$
\lim_{m \to \infty, \; \triangle_m \to 0} I(X_{t_1}^{t_m}; Y_{t_1}^{t_m})=I(X_0^T; Y_0^T),
$$
where $X_{t_1}^{t_m}, Y_{t_1}^{t_m}$ are samples of $X_0^T, Y_0^T$ at times $t_1, t_2, \ldots, t_m$.
\end{conj}

For band-limited white Gaussian non-feedback channels, the celebrated Shannon's sampling theorem~\cite{sh49} states that a sampling fine enough will completely determine the original channels. Naturally connecting discrete-time and continuous-time Gaussian channels in a more general setting, a valid Conjecture~\ref{discrete-to-continuous} would be of fundamental importance for a deeper understanding of continuous-time Gaussian feedback channels: it will provide valuable insights to the study of continuous-time channels in general, and may even allow straightforward translations of the results or techniques from the discrete-time setting to the continuous-time one for some special cases. Here, we remark that for continuous-time Gaussian feedback channels, a ``blind'' application of the discrete Fourier transform as in Shannon's treatment of white Gaussian non-feedback channels would be problematic since such a transform will render a loss of causality intrinsically residing in the original continuous-time channel.

With regard to Conjecture~\ref{discrete-to-continuous}, some progress has recently appeared in~\cite{li14}: under Novikov's condition, it has been shown that the sampling theorem as in the conjecture holds with respect to {\it increasingly refined} samplings; and furthermore, for the degenerate case that the feedback is absent, the techniques in the proof of Theorems~\ref{discrete-time-main-result}, together with Doob's martingale convergence theorem, has been be used to prove a sampling theorem for the MMSE, which, together with the continuous-time I-MMSE relationship, immediately implies Conjecture~\ref{discrete-to-continuous} as a corollary.

Among many possible applications, a valid Conjecture~\ref{discrete-to-continuous} may lead to a rigorous definition of continuous-time directed information. Here, we remark that the definition as in~\cite{we13}, which employs partitions of time intervals, is only valid for Gaussian channels with strictly delayed feedback. Conjecture~\ref{discrete-to-continuous} suggests a definition using samplings is more plausible for the general case.

Conjecture~\ref{discrete-to-continuous} may give us more insights on stationary Gaussian channels. Consider the following stationary Gaussian channel
\begin{equation}
Y(t)=X(t)+Z(t),
\end{equation}
where the noise $\{Z(t)\}$ is a stationary Gaussian process with spectral density function (SDF) $g(\lambda)$. Assume that the input $\{X(t)\}$ is also a stationary Gaussian process with SDF $f(\lambda)$. In this case, it has long been conjectured that the mutual information rate of the above channel can be computed as
\begin{equation}  \label{Pinsker-Conjecture}
\lim_{T \to \infty} \frac{1}{T} I(X_0^T; Y_0^T) = \frac{1}{4 \pi} \int_{-\infty}^{\infty} \log \left(1+\frac{f(\lambda)}{g(\lambda)}\right) d\lambda
\end{equation}
In some special cases, the above formula has been proved in a rigorous way. For example, if both $f(\lambda)$ and $g(\lambda)$ are rational SDF's, then the formula is true~\cite{Pinsker}. However, in general, there are some mathematical difficulties to prove (\ref{Pinsker-Conjecture}) rigorously. Coupled with the well-known fact that the counterpart result of (\ref{Pinsker-Conjecture}) in discrete time has been proved~\cite{ih93}, Conjecture~\ref{discrete-to-continuous} may help us to establish (\ref{Pinsker-Conjecture}) in full generality. Note that the proven discrete-time counterpart of (\ref{Pinsker-Conjecture}) has been adapted and used in Kim's characterization~\cite{ki10} of discrete-time Gaussian feedback capacity as the solution to a variational problem. So, it is conceivable that a proven (\ref{Pinsker-Conjecture}) will greatly enhance our understanding of continuous-time stationary Gaussian feedback channels.

\bigskip

{\bf Acknowledgement.} We would like to thank Dongning Guo, Young-Han Kim, Tsachy Weissman and Yihong Wu for insightful suggestions and comments, and for pointing out relevant references.

\section*{Appendices} \appendix

\section{Key Lemmas}  \label{key-lemmas}
The following two well-known lemmas are the main tools that will be used to justify the interchanges between a differentiation and an integration in this paper; for their proofs, see~\cite[Theorem A.5.1, Theorem A.5.2]{Durrett}.

\begin{lem} \label{interchange}
Let $f(x,\theta)$ be a continuously differentiable function with respect to $\theta$ and $X$ be a random variable. Let $\varepsilon>0$ and suppose that\\
(i) $u(\theta)=\E[f(X,\theta)]<\infty$ for all $\theta\in(\theta_0-\varepsilon,\theta_0+\varepsilon)$, and\\
(ii) $v(\theta)=\E[\frac{\partial}{\partial \theta}f(X,\theta)]$ is continuous at $\theta=\theta_0$, and\\
(iii) $\E\left(\int_{\theta_0-\varepsilon}^{\theta_0+\varepsilon}\left|\frac{\partial}{\partial \theta}f(X,\theta)\right|d \theta\right)<\infty$,\\
then we have $u'(\theta_0)=v(\theta_0)$, i.e.,
$$\left.\frac{d}{d\theta}\E[ f(X,\theta)]\right|_{\theta=\theta_0}=\left.\E\left[ \frac{\partial}{\partial \theta}f(X,\theta)\right]\right|_{\theta=\theta_0}.$$
\end{lem}

The following lemma is a direct consequence of the above one.
\begin{lem} \label{interchange'}
Let $f(x,\theta)$ be a continuously differentiable function with respect to $\theta$ and $X$ be a random variable. Let $\varepsilon>0$ and suppose that\\
(i) $u(\theta)=\E[f(X,\theta)]<\infty$ for $\theta\in(\theta_0-\varepsilon,\theta_0+\varepsilon)$, and \\
(ii) $\E\left[\sup\limits_{\theta \in(\theta_0-\varepsilon, \theta_0+\varepsilon)}\left|\frac{\partial}{\partial \theta}f(X, \theta)\right|\right]<\infty$,\\
then we have $u'(\theta_0)=v(\theta_0)$, i.e.,
$$\left.\frac{d}{d\theta}\E[ f(X,\theta)]\right|_{\theta=\theta_0}=\left.\E\left[ \frac{\partial}{\partial \theta}f(X,\theta)\right]\right|_{\theta=\theta_0}.$$
\end{lem}

\section{Justifications for the interchanges in Section~\ref{new-proof-1}}  \label{Uniform-Integrability}

{\bf Justification of $(a)$.} We will need to show that for any $\rho_0 \in \mathbb{R}$,
\begin{equation} \label{just-2}
\left. \frac{d}{d\rho} \E[\log f_{Y}(Y)] \right|_{\rho = \rho_0} = \left. \E \left[ \frac{d}{d\rho} \log f_{Y}(Y) \right] \right|_{\rho=\rho_0};
\end{equation}
and as will be done in other justifications in the sequel, we will check the technical conditions in the key lemmas in Section~\ref{key-lemmas}. Note that, by the assumption that $E[X^2] < \infty$, we have
\begin{equation*}
E[Y^2]=\rho^2 E[X^2]+E[N^2] < \infty.
\end{equation*}
On the other hand, it follows from
\begin{equation*}
H(Y) \geq H(Y|X)=H(Z)
\end{equation*}
that $H(Y)$ is lower bounded, which yields the finiteness of $\E[\log f(Y_1^n)]$ for all $\rho \in (\rho-\varepsilon, \rho+\varepsilon)$. As in the proof of Section~\ref{new-proof-1}, we have
\begin{equation*}
\E \left[ \frac{d}{d\rho} \log f_{Y}(Y) \right]= \rho \E[(X-\E[X|Y])^2]=\rho(\E[X^2]-\E[\E^2[X|Y]]),
\end{equation*}
which means to prove the continuity of $\E \left[ \frac{d}{d\rho} \log f_{Y}(Y) \right]$ at $\rho=\rho_0$, it suffices to prove that of $\E[\E^2[X|Y]]$ at $\rho=\rho_0$.

As a matter of fact, we will prove the aforementioned continuity at any $\rho$. We first show that
\begin{equation*}
\E[X|Y]=\frac{1}{f_Y(Y)} \int_{\mathbb{R}} \frac{x}{\sqrt{2 \pi}} e^{-(Y-\rho x)^2/2} f_X(x) dx
\end{equation*}
is continuous in $\rho$. To see this, note that for any $\rho$, we have
\begin{equation*}
\frac{x}{\sqrt{2 \pi}} e^{-(Y-\rho x)^2/2} f_X(x) \leq \frac{|x|}{\sqrt{2 \pi}} f_X(x),
\end{equation*}
of which the right hand side is integrable. It then follows from the fact that
$\frac{x}{\sqrt{2 \pi}} e^{-(Y-\rho x)^2/2} f_X(x)$ is continuous at any $\rho$ and the dominated convergence theorem that
\begin{equation*}
\int_{\mathbb{R}} \frac{x}{\sqrt{2 \pi}} e^{-(Y-\rho x)^2/2} f_X(x) dx
\end{equation*}
is continuous in $\rho$. A similar argument can be applied to show that $f_Y(Y)$ is also continuous in $\rho$, which immediately implies the continuity of $\E[X|Y]$ in $\rho$.

We are now ready to show that $\E[\E^2[X|Y]]$ is continuous in $\rho$. To see this, note that it follows from $\E[X^2] < \infty$ that $\{\E[X^2|Y], \, \rho \ge 0\}$ forms a family of uniformly integrable random variables. This, together with the fact that $\E^2[X|Y] \leq \E[X^2|Y]$, implies that $\{\E^2[X|Y], \, \rho \ge 0\}$ also forms a collection of uniformly integrable random variables. By Theorem $5.5.2$ in~\cite{Durrett}, the continuity of $\E[\E^2[X|Y]]$ then follows from that of $\E[X|Y]$ and the uniform integrability of $\{\E^2[X|Y], \rho \geq 0\}$.

Moreover, it can be readily verified that
\begin{align*}
\E \left[ \int_{\rho_0-\eps}^{\rho_0+\eps} \left|\frac{d}{d\rho} \log f_{Y}(Y) \right| d \rho \right]  & = \E \left[ \int_{\rho_0-\eps}^{\rho_0+\eps} \left|\int_{\R} (Y-\rho x) (X-x) f_{X|Y}(x|Y) dx \right| d \rho \right] \\
& \leq  \E \left[ \int_{\rho_0-\eps}^{\rho_0+\eps} \int_{\R} \left| (Y-\rho x) (X-x) \right| f_{X|Y}(x|Y) dx d \rho \right]\\
& \leq  \E \left[ \int_{\rho_0-\eps}^{\rho_0+\eps} \int_{\R} (|Y X|+|Y x|+\rho|x X|+\rho |x^2|) f_{X|Y}(x|Y) dx d \rho \right]\\
& = \int_{\rho_0-\eps}^{\rho_0+\eps} \E[\E[|YX|]+|Y|\E[|X||Y]+\rho |X| \E[|X||Y] + \rho \E[X^2|Y]] d \rho\\
&= \int_{\rho_0-\eps}^{\rho_0+\eps} 2\E[|YX|]+\rho \E[\E^2[|X||Y]]+\rho \E[X^2] d\rho\\
&\leq \int_{\rho_0-\eps}^{\rho_0+\eps} \rho \E[X^2] + \frac{1}{2} \E[X^2]+ \frac{1}{2} \E[N^2] + 2 \rho \E[X^2] d\rho,
\end{align*}
which is finite due to the assumption that $\E[X^2] < \infty$ and the fact that $\E[N^2] < \infty$. So, by Lemma \ref{interchange}, we can switch the integration and differentiation as in (\ref{just-2}).

{\bf Justification of $(b)$.} We first prove that for any $\rho_0 \in \mathbb{R}$,
\begin{equation*}
\left. \dfrac{d}{d\rho} \int_{\R}  f_{Y|X}(Y|x) f_X(x) dx \right|_{\rho=\rho_0}= \left. \int_{\R} \dfrac{d}{d\rho} f_{Y|X}(Y|x) f_X(x) dx \right|_{\rho=\rho_0},
\end{equation*}
or equivalently, we prove that for any $\rho_0 \in \mathbb{R}$ and for any $x', z' \in \mathbb{R}$,
\begin{equation}  \label{just-1}
\left. \dfrac{d}{d\rho} \int_{\R}  f_{Y|X}(\rho x' + z'|x) f_X(x) dx \right|_{\rho=\rho_0}= \left. \int_{\R} \dfrac{d}{d\rho} f_{Y|X}(\rho x'+z'|x) f_X(x) dx \right|_{\rho=\rho_0}.
\end{equation}

In what follows, fix $x', z' \in \mathbb{R}$ and $\varepsilon > 0$. Straightforward computations yield that for all $\rho \in (\rho_0-\varepsilon, \rho_0+\varepsilon)$
\begin{equation*}
\int_{\R}  f_{Y|X}(\rho x'+z'|x) f_X(x) dx \leq \frac{1}{\sqrt{2 \pi}} \int_{\R}   f_X(x) dx \leq \frac{1}{\sqrt{2 \pi}},
\end{equation*}
and moreover,
\begin{align*}
\dfrac{\partial}{\partial\rho} f_{Y|X}(\rho x'+z'|x)&=\frac{1}{\sqrt{2\pi}} \dfrac{\partial}{\partial\rho} \left[e^{-(\rho x'- \rho x+z')^2/2}\right]\\
&=-\frac{1}{\sqrt{2\pi}} e^{-(\rho x'- \rho x+z')^2/2} (\rho x' - \rho x + z') (x'-x),
\end{align*}
which, together with the assumption that $E[X^2] < \infty$, immediately implies that
\begin{equation*}
\int_{\R} \sup_{\rho \in (\rho_0-\varepsilon, \rho_0+\varepsilon)} \left| \frac{\partial}{\partial \rho} f_{Y|X}(\rho x'+z'|x) f_X(x) \right| dx < \infty.
\end{equation*}
The interchange as in (\ref{just-1}) then immediately follows from an invocation of Lemma~\ref{interchange'}.

\section{Justifications for the interchanges in Section~\ref{new-proof-2}} \label{deBruijn}

{\bf Justification of $(a)$.} We need to verify that for any $t_0 > 0$,
\begin{equation*}
\left. \dfrac{d}{dt} \int_{\R}  f_{Y|X}(Y|x) f_X(x) dx \right|_{t=t_0}= \left. \int_{\R} \dfrac{d}{dt} f_{Y|X}(Y|x) f_X(x) dx \right|_{t=t_0},
\end{equation*}
or equivalently, we prove that for any $t_0 >0$ and for any $x', z' \in \mathbb{R}$,
\begin{equation*}
\left. \dfrac{d}{dt} \int_{\R}  f_{Y|X}(x' + \sqrt{t} z'|x) f_X(x) dx \right|_{t=t_0}= \left. \int_{\R} \dfrac{d}{dt} f_{Y|X}(x'+\sqrt{t} z'|x) f_X(x) dx \right|_{t=t_0},
\end{equation*}
which follows from a parallel argument as in the proof of (\ref{just-1}).

{\bf Justification of $(b)$.} We need to verify that for any $t_0 > 0$,
\begin{equation*}
\left. \frac{d}{dt} \E[\log f_{Y}(Y)] \right|_{t = t_0} = \left. \E \left[ \frac{d}{dt} \log f_{Y}(Y) \right] \right|_{t=t_0},
\end{equation*}
which follows from a parallel argument as in the proof of (\ref{just-2}).

\section{Justifications for the interchanges in the Proof of Theorem~\ref{discrete-time-main-result}} \label{discrete-interchange}

In this section, we fix $\eps > 0$ and we sometimes write $g_i(W_1^i, Y_1^{i-1})$ as $g_i$ for notational simplicity.

{\bf Justification of $(a)$.} We need to prove that for any $\rho_0 \in \mathbb{R}$,
\begin{equation}  \label{just-4}
\left. \frac{d}{d\rho} \E[\log f(Y_1^n)] \right|_{\rho = \rho_0} = \left. \E \left[ \frac{d}{d\rho} \log f(Y_1^n) \right] \right|_{\rho=\rho_0}.
\end{equation}
Note that, by (\ref{discrete-condition-d}), we have for all $\rho \in [\rho_0-\eps, \rho_0+\eps]$ and for all $i$,
\begin{equation*}
E[Y_i^2]=\rho^2 E[\sup_{\rho \in [\rho_0-\eps, \rho_0+\eps]} g_i^2(W_1^i, Y_1^{i-1})]+E[Z_i^2] < \infty,
\end{equation*}
which implies that $H(Y_i)$ is upper bounded. On the other hand, it follows from
\begin{equation*}
H(Y_i) \geq H(Y_i|Y_1^{i-1}, W_i)=H(Z_i)
\end{equation*}
that $H(Y_i)$ is lower bounded, and so we have obtained the finiteness of $\E[\log f(Y_1^n)]$. As in the proof of Theorem~\ref{discrete-time-main-result}, we have
{\small \begin{align*}
\hspace{-1.5cm} \E \left[ \frac{d}{d\rho} \log f(Y_1^n) \right] & = -\rho \sum_{i=1}^n \E\left[(g_i-\E[g_i|Y_1^n])^2\right] - \rho \sum_{i=1}^n \E\left[Y_i \left(\frac{d}{d\rho} g_i-\left[\frac{d}{d\rho} g_i\right] \right) \right] - \rho^2 \sum_{i=1}^n \E\left[g_i \left[\frac{d}{d\rho} g_i \right]-\E[g_i|Y_1^n] \frac{d}{d\rho} g_i\right],
\end{align*}}
So, to prove the continuity of $\E \left[ \frac{d}{d\rho} \log f(Y_1^n) \right]$ at $\rho=\rho_0$, it suffices to prove that of
$$
\hspace{-1.5cm} \E[g_i^2(W_1^i, Y_1^{i-1})], \,\, \E[\E^2[g_i(W_1^i, Y_1^{i-1})|Y_1^n]], \,\, \E\left[Y_i \frac{d}{d\rho}g_i(W_1^i, Y_1^{i-1})\right], \,\, \E\left[Y_i \left[\frac{d}{d\rho}g_i\right](W_1^i, Y_1^{i-1})\right],
$$
\begin{equation} \label{six-quantities}
\E\left[g_i(W_1^i, Y_1^{i-1}) \left[\frac{d}{d \rho} g_i\right](W_1^i, Y_1^{i-1})\right], \,\,\E\left[\E[g_i(W_1^i, Y_1^{i-1})|Y_1^n] \frac{d}{d\rho} g_i(W_1^i, Y_1^{i-1})\right]
\end{equation}
at $\rho=\rho_0$. With Conditions (\ref{discrete-condition-d-1}) and (\ref{discrete-condition-d-2}) and the fact that for all feasible $i$, $g_i(W_1^i, Y_1^{i-1})$ is continuous in $\rho$, the continuity of $\E[g_i^2(W_1^i, Y_1^{i-1})]$ immediately follows from the dominated convergence theorem. Similarly, it can be also verified that
\begin{equation*}
\E \left[\sup_{\rho \in [\rho_0-\eps, \rho_0+\eps]} \left| g_i(W_1^i, Y_1^{i-1}) \frac{d}{d\rho} g_i(W_1^i, Y_1^{i-1}) \right| \right] < \infty,
\end{equation*}
which implies the continuity of $\E[g_i(W_1^i, Y_1^{i-1}) \frac{d}{d \rho} g_i(W_1^i, Y_1^{i-1})]$. Moreover, a similar argument as in Section~\ref{Uniform-Integrability} can be used to establish the continuity of other quantities in (\ref{six-quantities}) in $\rho$. We then obtain the continuity of $\E \left[ \displaystyle{\frac{d}{d\rho}} \log f(Y_1^n) \right]$, as desired.

Moreover, we verify that
\begin{align*}
\E \left[ \int_{\rho_0-\eps}^{\rho_0+\eps} \left|\frac{d}{d\rho} \log f(Y_1^n) \right| d \rho \right] &= \E\left[ \int_{\rho_0-\eps}^{\rho_0+\eps} \left|\int_{\R^n}\sum_{i=1}^n (Y_i-\rho g_i(w_1^i, Y_1^{i-1}) \bigg( g_i(W_1^i, Y_1^{i-1})-g_i(w_1^i, Y_1^{i-1}) \right. \right.\\
&\left. \left. \quad\quad \quad \quad\quad\quad  +\rho \frac{d}{d\rho} (g_i(W_1^i, Y_1^{i-1})- g_i(w_1^i, Y_1^{i-1}))\bigg) f(w_1^n|Y_1^n) dw_1^n \right| d\rho \right]\\
& \leq \E\left[ \int_{\rho_0-\eps}^{\rho_0+\eps} \int_{\R^n}\sum_{i=1}^n \left|(Y_i-\rho g_i(w_1^i, Y_1^{i-1}) \bigg( g_i(W_1^i, Y_1^{i-1})-g_i(w_1^i, Y_1^{i-1}) \right. \right.\\
&\left. \left. \quad\quad \quad \quad\quad\quad  +\rho \frac{d}{d\rho} (g_i(W_1^i, Y_1^{i-1})- g_i(w_1^i, Y_1^{i-1}))\bigg) \right| f(w_1^n|Y_1^n) dw_1^n  d\rho \right]\\
& < \infty,
\end{align*}
where the finiteness then follows from (\ref{discrete-condition-d-1}) and (\ref{discrete-condition-d-2}). So, by Lemma \ref{interchange}, the integration and differentiation in (\ref{just-4}) can be interchanged.

{\bf Justification of $(b)$.} We need to prove that for any $\rho_0 \in \mathbb{R}$, with probability $1$,
\begin{equation} \label{just-3}
\left. \frac{d}{d\rho}\int_{\R^n}  f(Y_1^n|w_1^n) f(w_1^n)dw_1^n \right|_{\rho=\rho_0} = \left. \int_{\R^n} \frac{d}{d\rho} f(Y_1^n|w_1^n) f(w_1^n)dw_1^n \right|_{\rho=\rho_0}.
\end{equation}
It follows from straightforward computations that for all $\rho \in (\rho_0-\eps, \rho_0+\eps)$
\begin{equation*}
\int_{\R}  f(Y_1^n|w_1^n) f(w_1^n) dw_1^n \leq \frac{1}{(\sqrt{2 \pi})^n} \int_{\R}   f(w_1^n) dw_1^n \leq \frac{1}{(\sqrt{2 \pi})^n}.
\end{equation*}
Moreover, we have
\begin{equation*}
\hspace{-1.5cm} \frac{d}{d \rho} f(Y_1^n|w_1^n) = \sum_{i=1}^n (Y_i-\rho g_i(w_1^i, Y_1^{i-1})) \bigg(g_i(W_1^i, Y_1^{i-1})-g_i(w_1^i, Y_1^{i-1})+\rho \frac{d}{d\rho} (g_i(W_1^i, Y_1^{i-1})- g_i(w_1^i, Y_1^{i-1}))\bigg) f(Y_1^n|w_1^n).
\end{equation*}
It then follows from (\ref{discrete-condition-d-1}) and (\ref{discrete-condition-d-2}) that
\begin{equation*}
\E\left[\int_{\R^n} \sup_{\rho \in [\rho_0-\eps, \rho_0+\eps]} \left|\frac{d}{d\rho} f(Y_1^n|w_1^n) f(w_1^n) \right| dw_1^n \right] < \infty,
\end{equation*}
which further implies that, with probability $1$,
\begin{equation*}
\int_{\R^n} \sup_{\rho \in [\rho_0-\eps, \rho_0+\eps]} \left|\frac{d}{d\rho} f(Y_1^n|w_1^n) f(w_1^n) \right| dw_1^n < \infty.
\end{equation*}
The interchange as in (\ref{just-3}) then immediately follows from an invocation of Lemma~\ref{interchange'}.

\section{Justifications for the interchanges in the Proof of Theorem~\ref{continuous-time-main-result-1}} \label{continuous-interchange}

In this section, let $\eps > 0$ and we sometimes write $g(s, W_0^s, Y_0^s)$ as $g(s)$ for notational simplicity.

{\bf Justification of $(a)$.} We need to prove that for any $\rho_0 \in \mathbb{R}$,
\begin{equation*}
\left. \frac{d}{d\rho} \int_0^T \E [g^2(s, W_0^s, Y_0^s)]ds \right|_{\rho=\rho_0} = \left. 2 \int_0^T \E \left[g(s, W_0^s, Y_0^s) \frac{d}{d\rho}g(s, W_0^s, Y_0^s) \right]ds \right|_{\rho=\rho_0}.
\end{equation*}

It immediately follows from Condition (d) that for any $\rho \in [\rho_0-\eps, \rho_0+\eps]$,
\begin{equation*}
\int_0^T \E [g^2(s, W_0^s, Y_0^s)]ds < \infty,
\end{equation*}
and moreover,
\begin{equation*}
\int_0^T \E\left[ \sup_{\rho \in [\rho_0-\eps, \rho_0+\eps]} 2 g(s, W_0^s, Y_0^s) \frac{d}{d \rho} g(s, W_0^s, Y_0^s) \right]
\end{equation*}
\begin{equation*}
\leq \int_0^T \E\left[\sup_{\rho \in [\rho_0-\eps, \rho_0+\eps]} g^2(s, W_0^s, Y_0^s) \right] ds+\int_0^T \E\left[\sup_{\rho \in [\rho_0-\eps, \rho_0+\eps]} \left(\frac{d}{d \rho}g(s, W_0^s, Y_0^s)\right)^2\right] ds < \infty.
\end{equation*}
The desired interchange then immediately follows from Lemma~\ref{interchange'}.

{\bf Justification of $(b)$.} We need to prove that for any $\rho_0 \in \mathbb{R}$, we have, with probability $1$,
\begin{equation*}
\left. \frac{d}{d\rho}\int \frac{d\mu_{Y|W}}{d\mu_B}(Y_0^T|w) \mu_W(dw) \right|_{\rho=\rho_0} = \left. \int  \frac{d}{d\rho}\frac{d\mu_{Y|W}}{d\mu_B}(Y_0^T|w) \mu_W(dw) \right|_{\rho=\rho_0}.
\end{equation*}

First of all, it follows from Theorem~\ref{Theorem7.1} that $\mu_Y \sim \mu_B$, and
\begin{equation*}
\frac{d\mu_{Y}}{d\mu_B}(Y_0^T)=\int \frac{d\mu_{Y|W}}{d\mu_B}(Y_0^T|w) \mu_W(dw)
\end{equation*}
is finite almost surely, which can be further written as
\begin{equation}  \label{continuity-by-integrability}
\frac{d\mu_{Y}}{d\mu_B}(Y_0^T) = \int \exp\left\{\rho^2 \int_0^T \tilde g(s) g(s) ds+\rho \int_0^T \tilde g(s) dB(s)-\frac{\rho^2}{2}\int_0^T \tilde g^2(s)ds\right\} \mu_W(dw),
\end{equation}
where $g(s, w_0^s, Y_0^s)$ is written as $\tilde{g}(s)$ for notational simplicity. Emphasizing the dependence on $\rho$, we write
$$
\quad b(\rho)= \exp\left\{\rho^2 \int_0^T \tilde g(s) g(s) ds+\rho \int_0^T \tilde g(s) dB(s)-\frac{\rho^2}{2}\int_0^T \tilde g^2(s)ds\right\},
$$
and write
\begin{equation*}
v(\rho) =\int \frac{d}{d\rho} b(\rho) \mu_W(dw)=\int b(\rho) c(\rho) \mu_W(dw),
\end{equation*}
where
\begin{align*}
c(\rho) & = \left(2\rho \int_0^T \tilde{g}(s) g(s) ds+ \rho^2 \int_0^T g(s) \frac{d}{d\rho}\tilde{g}(s) ds + \rho^2 \int_0^T \tilde{g}(s) \frac{d}{d\rho}g(s) ds \right.\\
& \left. + \int_0^T \tilde{g} dB(s)+\rho \int_0^T \frac{d}{d\rho} \tilde{g}(s) dB(s)-\rho \int_0^T \tilde{g}^2(s) ds -\rho^2 \int_0^T \tilde{g}(s) \frac{d}{d\rho} \tilde{g}(s) \right) .
\end{align*}
Now, we will show that there exists a constant $C$ such that for any $\rho_1 < \rho_2$,
$$
\E[|v(\rho_2)-v(\rho_1)|^2] \leq C |\rho_2-\rho_1|^2,
$$
which, by Kolmogov's continuity theorem~\cite{ka91}, implies the continuity of $v(\rho)$ (or, more precisely, $v(\rho)$ has a continuous modification). To this end, note that
\begin{align*}
v(\rho_2)-v(\rho_1) & = \int b(\rho_2) c(\rho_2) - b(\rho_1) c(\rho_1) \mu_W(dw)\\
& = \int (b(\rho_2)- b(\rho_1)) c(\rho_2) \mu_W(dw) - \int b(\rho_1) (c(\rho_2)-c(\rho_1)) \mu_W(dw)\\
& = \int \int_{\rho_1}^{\rho_2} b(\gamma) c(\gamma) d\gamma c(\rho_2) \mu_W(dw)- \int b(\rho_1) (c(\rho_2)-c(\rho_1)) \mu_W(dw),
\end{align*}
which, via tedious yet straightforward computations, can be expanded to a sum of expectation terms, each of which can be proven to be of $O((\rho_2-\rho_1)^2)$. In what follows, we only establish this for a couple of representative terms, since other terms can be handled in a parallel fashion.

{\bf Term $1$.} Letting $a(\rho)=\int_0^T g(s) \frac{d}{d\rho} \tilde{g}(s) ds$, we have
\begin{align*}
\hspace{-1cm} \E\left[\left|\int \int_{\rho_1}^{\rho_2} b(\gamma) a(\gamma) d\gamma a(\rho_2) \mu_W(dw)\right|^2 \right]
& \leq \E\left[\int \left| \int_{\rho_1}^{\rho_2} b(\gamma) a(\gamma) d\gamma a(\rho_2) \right|^2  \mu_W(dw) \right]\\
& =\E\left[\int \left| \int_{\rho_1}^{\rho_2} \int_{\rho_1}^{\rho_2} b(\gamma_1) b(\gamma_2) a(\gamma_1) a(\gamma_2) d\gamma_1 d\gamma_2 \right| a^2(\rho_2) \mu_W(dw) \right]\\
& \leq \int \int_{\rho_1}^{\rho_2} \int_{\rho_1}^{\rho_2} \E\left[\left| b(\gamma_1) b(\gamma_2) a(\gamma_1) a(\gamma_2) \right| \times \left|a(\rho_2)\right|^2 \right] d\gamma_1 d\gamma_2 \mu_W(dw)\\
&=O(|\rho_2-\rho_1|^2),
\end{align*}
where, for the last step, we have used (\ref{d-1}) in Condition (d) and the fact
$$
a(\rho) = \int_0^T g(s) \frac{d}{d\rho} \tilde{g}(s) ds \leq \int_0^T \frac{g^2(s)+(\frac{d}{d\rho}\tilde{g}(s))^2}{2} ds,
$$
and the well-known fact~\cite{ok95} that for any $K$,
\begin{equation} \label{any-K}
\E\left[\exp\left\{K \int_0^T \tilde{g}(s) dB(s)-K^2/2 \int_0^T \tilde{g}^2(s)ds\right\}\right] \leq 1.
\end{equation}

{\bf Term $2$.} We have
\begin{align*}
& \E \left| \int b(\rho_1) \left(\int_0^T \frac{d}{d\rho} \tilde{g}^{(\rho_1)} dB(s)-\int_0^T \frac{d}{d\rho} \tilde{g}^{(\rho_2)} dB(s) \right) \mu_W(dw) \right|^2\\
&=\E \left| \int b(\rho_1) \left(\int_0^T \left(\frac{d}{d\rho} \tilde{g}^{(\rho_1)}-\frac{d}{d\rho} \tilde{g}^{(\rho_2)}\right) dB(s)\right) \mu_W(dw) \right|^2\\
&\leq \int \E b^2(\rho_1) \left(\int_0^T \left(\frac{d}{d\rho} \tilde{g}^{(\rho_1)}-\frac{d}{d\rho} \tilde{g}^{(\rho_2)}\right) dB(s)\right)^2 \mu_W(dw)\\
&=O(|\rho_2-\rho_1|^2),
\end{align*}
where, for the last step, we have used (\ref{d-2}) in Condition (d) and (\ref{any-K}) and the Burkholder-Davis-Gundy inequality~\cite{ka91}.

After handling other terms in a similar fashion, the desired continuity is then established. Moreover, with Conditions (d) and (e), it is straightforward to verify that
\begin{equation*}
\E\left[\int_{\rho_0-\eps}^{\rho_0+\eps} \left|\frac{d}{d\rho} b(\rho)\right| d\rho \right] < \infty,
\end{equation*}
So, with all the technical conditions checked, the desired interchange then immediately follows from Lemma~\ref{interchange}.

{\bf Justification of $(c)$.} We need to prove that for any $\rho_0 \in \mathbb{R}$,
\begin{equation*}
\left. \dfrac{d}{d\rho}\E\left[\log \frac{d\mu_{Y}}{d\mu_B}(Y_0^T)\right] \right|_{\rho=\rho_0} = \left. \E\left[ \dfrac{d}{d\rho}\log \frac{d\mu_{Y}}{d\mu_B}(Y_0^T)\right] \right|_{\rho=\rho_0}.
\end{equation*}

First of all, we will show that for all $\rho \in [\rho_0-\eps, \rho_0+\eps]$, $\E\left[\log \frac{d\mu_{Y}}{d\mu_B}(Y_0^T)\right]$ is finite. To see this, first note that it follows from Theorem~\ref{Theorem7.1} that
\begin{equation*}
\frac{d\mu_{Y}}{d\mu_B}(Y_0^T)=\frac{1}{\E[e^{-\int_0^T g(s) dY+1/2 \int_0^T g^2(s) ds}|Y_0^T]}.
\end{equation*}
By Jensen's inequality, we have
\begin{equation*}
\E\left[ -\int^T_0 g(s) dY_s+\frac{1}{2}\int^T_0 g^2(s)ds |Y_0^T \right] \leq \log\E[e^{-\int^T_0 g(s)
dY_s+\frac{1}{2}\int^T_0 g^2(s)ds}|Y_0^T],
\end{equation*}
and, by the easy fact that $\log x \leq x$ for any $x > 0$,
\begin{equation*}
\log\E[e^{-\int^T_0 g(s) dY_s+\frac{1}{2}\int^T_0 g^2(s)ds}|Y_0^T] \leq \E[e^{-\int^T_0 g(s) dY_s+\frac{1}{2}\int^T_0
g^2(s)ds}|Y_0^T],
\end{equation*}
The desired finiteness then follows from
\begin{equation*}
\hspace{-1cm} \left|\log\E[e^{-\int^T_0 g(s) dY_s+\frac{1}{2}\int^T_0 g^2(s)ds}|Y_0^T]\right| \leq
\left|\E[ -\int^T_0 g(s) dY_s+\frac{1}{2}\int^T_0 g^2(s) ds |Y_0^T]\right|+\E[e^{-\int^T_0 g(s) dY_s+\frac{1}{2}\int^T_0 g^2(s)ds}|Y_0^T].
\end{equation*}

Next, as in the proof of Theorem~\ref{continuous-time-main-result-1}, we have
\begin{align*}
\E\left[\dfrac{d}{d\rho}\log \frac{d\mu_{Y}}{d\mu_B}(Y_0^T)\right] & = \rho \int_0^T \E[\E^2[g(s)|Y_0^T]] ds+\rho \E\left[\left.\E_W\left[\int_0^T \frac{d}{d\rho} g(s) dY(s)\right|Y_0^T \right]\right]\\
&+\rho^2 \int_0^T \E[\frac{d}{d\rho} g(s) \E[g(s)|Y_0^T]] ds-\rho^2 \int_0^T \E\left[ g(s) \E_W\left[\left.\frac{d}{d\rho} g(s)\right|Y_0^T \right] \right]ds.
\end{align*}
Note that
\begin{equation*}
\int_0^T \sup_{\rho \in [\rho_0-\eps, \rho_0+\eps]} \E[\E^2[g(s)|Y_0^T]]ds \leq \int_0^T \sup_{\rho \in [\rho_0-\eps, \rho+\eps]} \E[\E[g^2(s)|Y_0^T]]ds = \int_0^T \sup_{\rho \in [\rho_0-\eps, \rho+\eps]} \E[g^2(s)]ds < \infty,
\end{equation*}
and furthermore,
{\small \begin{equation*}
\hspace{-1.5cm} \int_0^T \sup_{\rho \in [\rho_0-\eps, \rho_0+\eps]} \E\left[\left|\E[g(s)|Y_0^T] \frac{d}{d\rho} g(s)\right|\right]ds \leq \frac{1}{2} \left( \int_0^T \sup_{\rho \in [\rho_0-\eps, \rho+\eps]} \E\left[ \E^2[g(s)|Y_0^T] \right] +  \sup_{\rho \in [\rho_0-\eps, \rho+\eps]} \E\left[ \left(\frac{d}{d\rho} g(s) \right)^2\right]ds \right) < \infty.
\end{equation*}}
In a similar fashion, one can establish that
$$
\E\left[\sup_{\rho \in [\rho_0-\eps, \rho_0+\eps]} \left.\E_W\left[\left|\int_0^T \frac{d}{d\rho} g(s) dY(s) \right| \right|Y_0^T \right]\right] < \infty,
$$
and
$$
\int_0^T \sup_{\rho \in [\rho_0-\eps, \rho_0+\eps]} \E\left[ \left| g(s) \E_W\left[\left.\frac{d}{d\rho} g(s)\right|Y_0^T \right] \right| \right]ds < \infty.
$$
It then immediately follows that $\E\left[\dfrac{d}{d\rho}\log \frac{d\mu_{Y}}{d\mu_B}(Y_0^T)\right]$ is continuous with respect to $\rho$. Moreover, note that
{\small \begin{align*}
\int_{\rho_0-\eps}^{\rho_0+\eps} \E\left[ \left|\dfrac{d}{d\rho}\log \frac{d\mu_{Y}}{d\mu_B}(Y_0^T) \right| \right] d\rho
&=\int_{\rho_0-\eps}^{\rho_0+\eps} \E\left[\left| \dfrac{d}{d\rho}\left(\frac{d\mu_{Y}}{d\mu_B}(Y_0^T)\right)/\frac{d\mu_{Y}}{d\mu_{B}}(Y_0^T) \right| \right] d\rho\\
& \hspace{-6cm} \leq \int_{\rho_0-\eps}^{\rho_0+\eps} \E  \bigg[  \E \left[\left. \left| \int_0^T g(s) dY(s)\right| \right|Y_0^T \right] +\rho \int \left| \int_0^T \frac{d}{d\rho} \tilde{g}(s) dY(s) \right| \mu_{W|Y}(dw|Y_0^T)  \\
& \hspace{-7cm}\quad +\rho\int_0^T \left|\E[g(s)|Y_0^T]g(s)-\E[g^2(s)|Y_0^T]\right|ds+\rho^2 \int_0^T \left| \frac{d}{d\rho} g(s) \E[g(s)|Y_0^T]\right| ds+\rho^2 \int_0^T \left|\tilde{g}(s) \frac{d}{d\rho} \tilde{g}(s)\right| ds \mu_{W|Y}(dw|Y_0^T) \bigg] d\rho.
\end{align*}}
It then follows from Condition (d) that
\begin{equation*}
\int_{\rho_0-\eps}^{\rho_0+\eps} \E\left[ \left|\dfrac{d}{d\rho}\log \frac{d\mu_{Y}}{d\mu_B}(Y_0^T) \right| \right] d\rho < \infty.
\end{equation*}
Finally, with all the technical conditions checked, the desired interchange follows from Lemma~\ref{interchange}.

\end{document}